\documentclass[11pt]{article}
\usepackage{fullpage,amsfonts,amsmath,amsthm,amssymb,mathtools,mathrsfs,graphicx,upgreek}
\usepackage[left=2.54cm,top=2.54cm,right=2.54cm,bottom=2.54cm]{geometry}

\usepackage{epsfig}
\usepackage{bm}

\usepackage[utf8]{inputenc}
\usepackage[T1]{fontenc}
\usepackage{epsfig}
\usepackage{bm}
\usepackage{comment}
\usepackage{hyperref}

\usepackage{bbm}
\usepackage{algorithm}
\usepackage{algpseudocode}

\usepackage{datetime}
\usepackage{hyperref}

\usepackage{enumitem}
\newcommand{\subscript}[2]{$#1 _ #2$}

\numberwithin{equation}{section}

\newtheorem{theorem}{Theorem}[section]

\newtheorem{proposition}[theorem]{Proposition}

\newtheorem{lemma}[theorem]{Lemma}

\theoremstyle{definition}
\newtheorem{example}[theorem]{Example}
\newtheorem{remark}[theorem]{Remark}

\newcommand{\st}{\text{st}}
\newcommand{\sta}{\text{st}}
\newcommand{\scr}{\mathcal}
\newcommand{\mb}{\mathbb}
\newcommand{\til}{\widetilde}

\newcommand{\eps}{\varepsilon}

\newcommand{\val}{\text{val}}
\newcommand{\OPT}{\text{OPT}}
\newcommand{\nOPT}{\OPT_{\text{non}}}
\newcommand{\rOPT}{\OPT_{\text{rel}}}
\newcommand{\LPOPT}{\text{LPOPT}}

\newcommand{\dLPOPT}{\LPOPT_{\text{DP}}}
\newcommand{\nLPOPT}{\LPOPT_{\text{DP-non}}}
\newcommand{\cLPOPT}{\LPOPT_{\text{conf}}}

\usepackage[usenames,dvipsnames]{xcolor}

\begin{document}

\title{Greedy Approaches to Online Stochastic Matching}

\author{Allan Borodin
\thanks{Department of Computer Science, University of Toronto, Toronto, ON, Canada
\texttt{bor@cs.toronto.edu}}
\and
Calum MacRury
\thanks{Department of Computer Science, University of Toronto, Toronto, ON, Canada
\texttt{cmacrury@cs.toronto.edu}}
\and
Akash Rakheja
\thanks{Department of Computer Science, University of Toronto, Toronto, ON, Canada
\texttt{rakhejaakash@gmail.com}}
}

\date{}
\maketitle

\begin{abstract}

Within the context of stochastic probing with commitment, we consider
the online stochastic matching problem; that is, the one-sided online bipartite
matching problem where edges adjacent to an online node must be probed to
determine if they exist based on edge probabilities that become known when an online vertex arrives.  If a probed edge
exists, it must be used in the matching (if possible). We consider the
competitiveness of online algorithms in both the adversarial order model (AOM)
and the random order model (ROM). More specifically, we consider a bipartite stochastic graph  
$G = (U,V,E)$ where $U$ is the set of offline vertices, $V$ is the set of online vertices and $G$
has edge probabilities $(p_{e})_{e \in E}$ and edge weights $(w_{e})_{e \in E}$.
Additionally, $G$ has probing constraints $(\scr{C}_{v})_{v \in V}$,  where $\scr{C}_v$
indicates which sequences of edges adjacent to an online vertex $v$ can be probed. We assume that $U$ is
known in advance, and that $\scr{C}_v$, together with the
edge probabilities and weights adjacent to an online vertex  are only revealed  when the online vertex arrives. 
This model generalizes the various settings of the classical bipartite matching problem,
and so our main contribution is in making progress towards understanding which classical results extend to the stochastic
probing model.

\end{abstract}

\section{Introduction}

Stochastic probing problems are part of the larger area of decision making under uncertainty and more specifically, stochastic optimization. Unlike more standard forms of stochastic optimization, it is not just that there is some stochastic uncertainty in the set of inputs, stochastic probing problems involve inputs that cannot be determined without probing (at some cost and/or within some constraint) so as to reveal the inputs. Applications of stochastic probing occur naturally in many settings, such as in matching problems where compatibility (for example, in online dating and kidney exchange applications) or legality (for example, a financial transaction that must be authorized before it can be completed) cannot be determined without some trial or investigation.       
Amongst other applications, the online bipartite stochastic matching problem notably models online advertising where the probability of an edge can correspond to the probability of a  purchase in online stores or to pay-per-click revenue in online searching.

The {\it(offline) stochastic matching} problem was introduced by Chen et al.~\cite{Chen}. In this problem, we are given an adversarially generated {\it stochastic graph} $G = (V, E)$ with a probability $p_e$ associated with each edge $e$ and a patience (or time-out) parameter $\ell_v$ associated with each vertex $v$. An algorithm probes edges in $E$ within the constraint that at most $\ell_v$ edges are probed incident to any particular vertex $v \in V$. Also, when an edge $e$ is probed, it is guaranteed to exist with probability exactly $p_e$. If an edge $(u,v)$ is found to exist, it is added to the matching and then $u$ and $v$ are no longer available.
The goal is to maximize the expected size of a matching constructed in this way. This problem can be generalized to vertices or edges having weights. We shall refer to this setting as the {\it known stochastic graph} setting.

Mehta and Panigrahi \cite{MehtaP12} adapted the offline stochastic matching model to online bipartite matching as originally studied in the classical (non-stochastic) adversarial order online model. That is, they consider the setting where the stochastic graph is unknown and online vertices are determined by an adversary.  More specifically, they studied the problem in the case of an unweighted stochastic graph $G = (U, V, E)$ where $U$ is the set of offline vertices and the vertices  in $V$ arrive online without knowledge of future online node arrivals. They considered the special case of uniform edge probabilities (i.e, $p_{e} =p$ for all $e \in E$) and \textit{unit} patience values, that is $\ell_v =1$ for all $v \in V$. 
Mehta et al. \cite{Mehta2015} considered the unweighted online stochastic bipartite setting with  arbitrary edge probabilities,
attaining a competitive ratio of $0.534$, and recently, Huang and Zhang \cite{huang2020online} additionally handled the case of arbitrary offline vertex weights, while improving this ratio to $0.572$.
However, as in \cite{MehtaP12}, both \cite{Mehta2015} and  \cite{huang2020online} are restricted to unit patience values, and moreover require edge probabilities which are \textit{vanishingly small}\footnote{Vanishingly small edge probabilities must satisfy $\max_{e \in E} p_{e} \rightarrow 0$,
where the asymptotics are with respect to the size of $G$.}. Goyal and Udwani \cite{Goyal2020OnlineMW} improved on both of these works by showing a $0.596$ competitive ratio in the same setting.


In all our results we will assume {\it commitment}; that is, when an edge is probed and found to exist, it must be included in the matching (if possible without violating the matching constraint). The patience constraint can be viewed as a simple form of a budget constraint for the online vertices. For some, but not all of our settings, we generalize patience and budget constraints by associating a set $\scr{C}_v$ of probing sequences for each online node $v$ where $\scr{C}_v$
indicates which sequences of edges adjacent to  vertex $v$ can be probed.

For random order or adversarial order of online vertex arrivals, results for the online Mehta and Panigraphi model (even for unit patience) generalize the  corresponding classical  non-stochastic models where edges adjacent to an online node are known upon arrival and do not need to be probed. It follows that any in-approximations in  the classical setting apply to the corresponding  stochastic setting. Further generalizing the classical settings, when the stochastic graph is unknown, a competitive ratio for the random order model implies that the same ratio is obtained in the stochastic i.i.d. model (for an unknown distribution) as proven in the classical setting by Karande et al. \cite{KarandeMT11}.

In 
a related paper 
\cite{borodin2021prophet}, we consider the setting when the stochastic graph is unknown,
but there is a known stochastic type graph from which arrivals are drawn independently, thereby
generalizing the i.i.d. model introduced by Bansal et al. \cite{BansalGLMNR12}. 
Our generalization in \cite{borodin2021prophet} allows for independent (but not necessarily identical) distributions; that is, we have an i.d. model.

\subsection{Preliminaries}

An input to the \textbf{(online) stochastic matching problem} is a \textbf{(bipartite) stochastic graph}, specified in the following way. Let $G=(U,V,E)$ be a bipartite graph with edge weights $(w_{e})_{e \in E}$ and edge probabilities $(p_{e})_{e \in E}$.
We draw an independent Bernoulli random variable of parameter $p_{e}$
for each $e \in E$. We refer to this Bernoulli as the \textbf{state} of the edge $e$, and 
denote it by $\sta(e)$. If $\sta(e)=1$, then we say that $e$ is \textbf{active}, and otherwise we say that
$e$ is \textbf{inactive}. For each $v \in V$, suppose that $\partial(v)^{(*)}$ corresponds to the collection of strings (tuples) formed from \textit{distinct} edges of $\partial(v)$. Each $v \in V$ has a set of  \textbf{online probing constraint} $\scr{C}_v \subseteq \partial(v)^{(*)}$ which is \textbf{substring-closed}.
That is, $\scr{C}_v$ has the property that if $\bm{e} \in \scr{C}_v$,
then so is any substring of $\bm{e}$. 
Similarly, we say that $\scr{C}_v$ is \textbf{permutation-closed}, provided if $\bm{e} \in \scr{C}_v$,
then so is any permutation string of $\bm{e}$. Observe that if $\scr{C}_v$
is both substring-closed and permutation-closed, then it corresponds to a downward-closed
family of subsets of $\partial(v)$. Thus, in particular, our setting encodes the case when $v$ has a patience value $\ell_v$, and more generally, when $\scr{C}_v$ corresponds to a matroid or budgetary constraint on $\partial(v)$. 
Note that we will often assume $w.l.o.g.$ that $E=U \times V$, as we can always
set $p_{u,v}:= 0$.

A solution to 
online stochastic matching 
is an \textbf{online probing algorithm}. An online probing algorithm
is initially only aware of the identity of the offline vertices $U$ of $G$. We think of $V$, as well as the relevant edges probabilities, weights,
and probing constraints, as being generated by an adversary. An ordering on $V$ is then generated either through an adversarial process or uniformly at random. We refer to the
former case as the \textbf{adversarial order model (AOM)} and the latter case as the \textbf{random order model (ROM)}.

Based on whichever ordering is generated on $V$,
the nodes are then presented to the online probing algorithm one by one.
When an online node $v \in V$ arrives, the online probing algorithm sees all the adjacent edges and their associated probabilities,
as well as $\scr{C}_v$. However, the edge states $(\sta(e))_{e \in \partial(v)}$
remain hidden to the algorithm. Instead, the algorithm must perform a \textbf{probing operation} on an adjacent edge $e$ to reveal/expose its state, $\sta(e)$. Moreover, the online probing algorithm must \textbf{respect commitment}.
That is, if an edge $e = (u,v)$ is probed and turns out to be active, then $e$ must be added to the current matching, provided $u$ and $v$ are both currently unmatched. The probing constraint $\scr{C}_v$ of the online node then restricts
which sequences of probes can be made to $\partial(v)$. As in the classical problem, an online probing algorithm must decide on a possible match for an online node $v$ before seeing the next online node. The goal of the online probing algorithm 
is to return a matching whose expected weight is as large as possible. Since $\scr{C}_v$ may be exponentially large in the size of $U$, in order to
discuss the efficiency of an online probing algorithm,
we work in the \textbf{membership query model}.
That is, upon receiving the online vertex $v \in V$,
an online probing algorithm may make a \textbf{membership query} to any string $\bm{e} \in \partial(v)^{(*)}$,
thus determining in a single operation whether 
or not $\bm{e} \in \partial(v)^{(*)}$ is in $\scr{C}_v$.

It is easy to see we cannot hope to obtain a non-trivial
competitive ratio against the expected value of an optimum matching of the stochastic graph\footnote{Consider a single online vertex with patience $1$, and $n$ offline (unweighted) vertices where each edge $e$ has probability $\frac{1}{n}$ of being present. The expectation of an online probing algorithm will be at most $\frac{1}{n}$ while the expected size of an optimal matching (over all instantiations of the edge probabilities) will be $1 - (1-\frac{1}{n})^n \rightarrow 1 -\frac{1}{e}$. This example clearly shows that no constant ratio is possible if the patience is sublinear (in $n = |U|$).}.  
Thus, the standard in the literature is to instead benchmark the performance of an online
probing algorithm against an \textit{optimum offline probing algorithm}. An \textbf{offline probing algorithm}
knows $G=(U,V,E)$, but initially the edge states $(\sta(e))_{e \in E}$ are hidden. It can adaptively probe the edges of $E$ in any order, but must satisfy the probing constraints $(\scr{C}_v)_{v \in V}$ at each step of its execution\footnote{Edges $\bm{e} \in E^{(*)}$ may be probed in order, provided $\bm{e}^{v} \in \scr{C}_v$
for each $v \in V$, where $\bm{e}^{v}$ is the substring of $\bm{e}$ restricted to edges of $\partial(v)$.},
while respecting commitment; that is, if a probed edge $e=(u,v)$ turns out to be active,
then $e$ is added to the matching (if possible). The goal of an offline probing
algorithm is to construct a matching with optimum weight in expectation. We define
the \textbf{committal benchmark} as an optimum offline probing algorithm, and use
$\OPT(G)$ to denote the expected value of the matching the committal benchmark
constructs. We abuse notation slightly, and also use $\OPT(G)$ to refer to the \textit{strategy}
of the committal benchmark on $G$. In Appendix \ref{sec:non_committal_benchmark},
we introduce the stronger \textbf{non-committal benchmark},
and indicate which of our results hold against it.

\subsection{Our Results}

We first consider the unknown stochastic matching problem in
the most general setting of arbitrary edge weights, and substring-closed probing constraints.
Note that since no non-trivial competitive ratio can be proven in the case of adversarial
arrivals, we work in the ROM setting. We generalize the   
matching algorithm of Kesselheim et al. \cite{KRTV2013} so as to apply to the stochastic probing setting.

\begin{theorem}\label{thm:ROM_edge_weights}
Suppose the adversary presents an edge-weighted stochastic graph $G=(U,V,E)$, 
with substring-closed probing constraints $(\scr{C}_v)_{v \in V}$. If $\scr{M}$
is the matching returned by Algorithm \ref{alg:ROM_edge_weights} when executing on $G$,
then 
\[
\mb{E}[w(\scr{M})] \ge \left(\frac{1}{e} - \frac{1}{|V|} \right) \cdot \OPT(G),
\]
provided the vertices of $V$ arrive uniformly at random ($u.a.r.)$. If the constraints $(\scr{C}_v)_{v \in V}$
are also permutation-closed, then Algorithm \ref{alg:ROM_edge_weights}
can be implemented efficiently in the membership oracle model.

\end{theorem}
Upon receiving the online vertices $V_{t}:=\{v_1, \ldots ,v_{t}\}$, in order
to generalize the matching algorithm of Kesselheim et al. \cite{KRTV2013}, Algorithm \ref{alg:ROM_edge_weights}
would ideally probe the edges of $\partial(v_t)$ suggested by $\OPT(G_t)$, where $G_{t}:=G[U \cup V_t]$
is the \textbf{induced stochastic graph}\footnote{Given $L \subseteq U, R \subseteq  V$, the induced stochastic graph 
$G[L \cup R]$ is formed by restricting the edges weights and probabilities of $G$ to those edges within $L \times R$.
Similarly, each probing constraint $\scr{C}_v$ is restricted to those strings whose entries lie entirely
in $L \times \{v\}$.} on $U \cup V_t$. However, since we wish for our algorithms to
be efficient in addition to attaining optimum competitive ratios, this strategy is not feasible. Our solution is to instead solve a configuration LP (\ref{LP:config}) for $G_t$, which was recently introduced by the authors in \cite{borodin2021prophet}, and whose optimum value upper bounds $\OPT(G_t)$. Unlike \ref{LP:standard_definition_general} --
the most prevalent LP used in the stochastic matching literature, originally introduced by Bansal et al. \cite{BansalGLMNR12}
(see Appendix \ref{sec:committal_relaxation})-- \ref{LP:config} allows us to handle general probing constraints, while not overestimating the performance of $\OPT(G_t)$. Our algorithm differs from the classical algorithm of Kesselheim et al. in that it is randomized.

Since Theorem \ref{thm:ROM_edge_weights} achieves the optimum asymptotic competitive ratio
for edge weights, in order to improve this result we next consider the case when the stochastic
graph $G=(U,V,E)$ has \textbf{(offline) vertex weights} -- i.e.,
there exists $(w_u)_{u \in U}$ such that $w_{u,v} = w_{u}$ for each $v \in N(u)$.
We consider a \textit{greedy} online probing algorithm.
That is, upon the arrival of $v$,
the probes to $\partial(v)$
are made in such a way that $v$ gains as much value
as possible (in expectation), provided the currently unmatched nodes of $U$
are equal to $R$. As such, we must follow the probing strategy
of the committal benchmark when restricted to $G[ \{v\} \cup R]$,
which we denote by $\OPT(R,v)$ for convenience.

Observe that if $v$ has unit
patience, then $\OPT(R,v)$  reduces to probing the adjacent edge $(u,v) \in R \times \{v\}$ such that
the value $w_{u} \cdot p_{u,v}$ is maximized. Moreover, 
if $v$ has unlimited patience, then $\OPT(R,v)$ corresponds to probing
the adjacent edges of $R \times \{v\}$ in non-increasing order of the associated vertex weights. 
Building on a result in Purohit et al. \cite{Purohit2019}, Brubach et al. \cite{Brubach2019} showed how to devise an \textit{efficient} probing strategy for $v$ whose expected value matches $\OPT(R,v)$, no matter the patience constraint. Using this probing strategy, they devised an online probing algorithm which achieves a competitive ratio of $1/2$ for arbitrary 
patience values. We extend their result to substring-closed probing constraints.

\begin{theorem}\label{thm:adversarial}
Suppose the adversary presents a vertex weighted stochastic graph $G=(U,V,E)$, 
with substring-closed probing constraints $(\scr{C}_v)_{v \in V}$. If $\scr{M}$
is the matching returned by Algorithm \ref{alg:dynamical_program} when executing on $G$,
then 
\[
\mb{E}[w(\scr{M})] \ge \frac{1}{2} \cdot \OPT(G),
\]
provided the vertices of $V$ arrive in adversarial order. Moreover, Algorithm \ref{alg:dynamical_program}
can be implemented efficiently, provided the constraints $(\scr{C}_v)_{v \in V}$ are also permutation-closed.
\end{theorem}

Since Algorithm \ref{alg:dynamical_program} is deterministic, the $1/2$ competitive ratio
is best possible for deterministic algorithms in the adversarial arrival setting. 
One direction is thus to instead consider what can be done if the online probing
algorithm is allowed randomization. In the case of
unit patience, this is the previously discussed setting of \cite{MehtaP12}, in which Mehta and Panigrahi showed that $.621 < 1-\frac{1}{e}$ is a randomized in-approximation with regard to  guarantees made against \ref{LP:standard_benchmark}, even for the unweighted uniform probability setting. Note that for unit patience, \ref{LP:standard_benchmark} can be viewed as encoding a \textit{relaxed} optimum probing algorithm which need only match the offline nodes once in expectation, and thus
upper bounds/relaxes $\OPT(G)$ (see Appendix \ref{sec:LP_relaxation} for details). This hardness result led Goyal and Udwani \cite{Goyal2020OnlineMW} to consider a new unit patience LP that is a tighter relaxation of $\OPT(G)$ than \ref{LP:standard_benchmark}, thereby allowing them to prove a $1-1/e$ competitive ratio for the case of \textbf{vertex-decomposable}\footnote{Vertex-decomposable means that there exists probabilities $(p_{u})_{u \in U}$ and $(p_{v})_{v \in V}$, such that $p_{(u,v)}=p_{u} \cdot p_{v}$ for each $(u,v) \in E$.} edge probabilities. However, they also discuss the difficulty of extending this result to the case of general edge probabilities. Our next contribution is thus to consider Algorithm \ref{alg:dynamical_program} in the ROM setting with unit patience, where we show
these difficulties do \textit{not} arise. In fact, we show that a $1-1/e$ performance guarantee
is provable against \ref{LP:standard_benchmark}, which shows that the in-approximation of Mehta and Panigraphi
does not apply (even for deterministic probing algorithms).

Upon the arrival of vertex $v$, if $v$ has unit patience,
then Algorithm \ref{alg:dynamical_program} reduces
to probing the edge $e =(u,v) \in \partial(v)$ such that $u \in U$ is currently unmatched,
and for which $w_{u} \cdot p_{u,v}$ is
maximized. Let us refer to this special case of Algorithm \ref{alg:dynamical_program}
as \textsc{GreedyProbe}.
\begin{theorem} \label{thm:ROM_unit_patience}
Suppose the adversary presents a vertex weighted stochastic graph $G=(U,V,E)$, 
with unit patience values. If $\scr{M}$
is the matching returned by \textsc{GreedyProbe} when executing on $G$,
then $\mb{E}[w(\scr{M})] \ge \left(1- \frac{1}{e}\right) \cdot \OPT(G)$,
provided the vertices of $V$ arrive in random order.
\end{theorem}
\begin{remark}
The guarantee of Theorem \ref{thm:ROM_unit_patience}
is proven against \ref{LP:standard_benchmark}, and together with the $0.621$
inapproximation of Mehta and Panigraphi, implies that deterministic probing algorithms
in the ROM setting have strictly more power than randomized probing algorithms in
the adversarial order model. The analysis of \textsc{GreedyProbe} is tight, as an execution of \textsc{GreedyProbe} corresponds to the seminal Karp et al. \cite{KarpVV90} \textsc{Ranking} algorithm for unweighted non-stochastic (i.e., $p_e  \in \{0,1\}$ for all $e \in E$) bipartite  matching \footnote{In the classical (unweighted, non-stochastic) online matching problem, an execution of the randomized \textsc{Ranking} algorithm in the adversarial setting can be coupled with an execution of the deterministic greedy algorithm in the ROM setting -- the latter of which is a special case of \textsc{GreedyProbe}. The tightness of the ratio $1-1/e$ therefore follows since this ratio is tight for the \textsc{Ranking} algorithm.}. 
\end{remark}

Our final result, Theorem \ref{thm:ROM_rankable}, makes partial progress towards
understanding the ROM setting in the case of general probing constraints.
However, unlike the adversarial setting,
the complexity of the constraints greatly impacts what we are able to prove. We state and prove our result in Section \ref{sec:vertex_weights}, which we note subsumes Theorem \ref{thm:ROM_unit_patience}.

%


\section{Edge Weights}\label{sec:edge_weights}

Let us suppose that $G=(U,V,E)$ is a stochastic graph with arbitrary edge weights, probabilities
and constraints $(\scr{C}_v)_{v \in V}$. For each $k \ge 1$ and $\bm{e} =(e_{1}, \ldots , e_{k}) \in E^{(*)}$, define $g(\bm{e}) := \prod_{i=1}^{k} (1 - p_{e_i})$.
Notice that $g(\bm{e})$ corresponds to the probability that 
all the edges of $\bm{e}$ are inactive, where $g(\lambda):=1$
for the empty string $\lambda$. We
also define $\bm{e}_{< e_i} := (e_{1}, \ldots ,e_{i-1})$ for each $2 \le i \le k$,
which we denote by $\bm{e}_{< i}$ when clear. By convention, $\bm{e}_{< 1}:= \lambda$.
Observe then that $\val(\bm{e}):=\sum_{i=1}^{|\bm{e}|} p_{e_i} w_{e_i} \cdot g(\bm{e}_{< i})$
corresponds to the expected weight of the first active edge of $\bm{e}$ if $\bm{e}$ 
is probed in order of its indices. Finally, for each $v \in V$ and $\bm{e} \in \scr{C}_v$, we introduce a decision variable $x_{v}(\bm{e})$. We can then express the following LP from \cite{borodin2021prophet}:
\begin{align}\label{LP:config}
\tag{LP-config}
&\text{maximize} &  \sum_{v \in V} \sum_{\bm{e} \in \scr{C}_v } \val(\bm{e}) \cdot x_{v}(\bm{e}) \\
&\text{subject to} & \sum_{v \in V} \sum_{\substack{ \bm{e} \in \scr{C}_v: \\ (u,v) \in \bm{e}}} 
p_{u,v} \cdot g(\bm{e}_{< (u,v)}) \cdot x_{v}( \bm{e})  \leq 1 && \forall u \in U  \label{eqn:relaxation_efficiency_matching}\\
&& \sum_{\bm{e} \in \scr{C}_v} x_{v}(\bm{e}) = 1 && \forall v \in V,  \label{eqn:relaxation_efficiency_distribution} \\
&&x_{v}( \bm{e}) \ge 0 && \forall v \in V, \bm{e} \in \scr{C}_v
\end{align}
\begin{theorem}[Theorem $3.1$ in \cite{borodin2021prophet}]\label{thm:new_LP_relaxation}
For any stochastic graph $G=(U,V,E)$ with substring-closed probing
constraints, $\OPT(G) \le \LPOPT(G)$, where $\cLPOPT(G)$ is 
the optimum value of \ref{LP:config}. 
\end{theorem}
In order to prove Theorem \ref{thm:new_LP_relaxation},
the natural approach is to define $x_{v}(\bm{e})$ to be the probability that the committal benchmark
probes the edges of $\bm{e}$ in order, where $v \in V$ and $\bm{e} \in \scr{C}_v$. 
However, to our knowledge, this interpretation of the decision variables of \ref{LP:config}
does not seem to yield a proof of Theorem \ref{thm:new_LP_relaxation}, 
for technical reasons which we discuss in detail in Appendix \ref{sec:committal_relaxation}.

The approach we take in \cite{borodin2021prophet} is to instead introduce a \textbf{combinatorial
relaxation} of the committal benchmark, which is a new stochastic probing problem on $G$
whose optimum solution upper bounds $\OPT(G)$. Specifically, we introduce the \textbf{relaxed stochastic matching problem},
a solution to which we refer to as a \textbf{relaxed probing algorithm}. A relaxed probing
algorithm operates in the same probing framework as an offline probing algorithm, however it does \textit{not}
return a (two-sided) matching of $G$. Instead, it returns a subset of its active
probes which form a one-sided matching of $V$. This one-sided matching  $\scr{N}$ has the additional property
that each offline vertex of $U$ is included in $\scr{N}$ at most once \textit{in expectation}. We define
the \textbf{relaxed benchmark} to be an optimum relaxed probing algorithm, and denote its evaluation on $G$
by $\rOPT(G)$. Since each offline probing algorithm is a relaxed probing algorithm,
clearly $\OPT(G) \le \rOPT(G)$. On the other hand, by defining
$x_{v}(\bm{e})$ to be the probability the relaxed benchmark probes $\bm{e} \in \scr{C}_v$
of $v \in V$ in order, we can show that $\rOPT(G) = \LPOPT(G)$, which implies Theorem \ref{thm:new_LP_relaxation}. 
For completeness, we include the details in Appendix \ref{sec:committal_relaxation}.

Not only does \ref{LP:config} relax the committal benchmark,
it can be solved efficiently, provided the constraints
$(\scr{C}_v)_{v \in V}$ are assumed to be closed under substrings
and permutations. The approach we take in \cite{borodin2021prophet}
is to first consider the dual of \ref{LP:config}. It is not hard to verify
that the \textsc{DP-OPT} algorithm of Theorem \ref{thm:efficient_star_dp}
can be used as a (deterministic) polynomial time separation oracle for this LP.
This ensures that the dual of \ref{LP:config} can be solved efficiently, as a consequence of how the ellipsoid algorithm \cite{Groetschel,GartnerM} executes. Moreover, by tracking which polynomial
number of constraints of the dual of \ref{LP:config} are queried by this separation oracle, one can reduce
the number of decision variables needed in \ref{LP:config} to a polynomial number. Since this restriction
of \ref{LP:config} has a polynomial number of constraints, it can then be solved efficiently. We omit
the details, as a more complete proof is given in \cite{borodin2021prophet},
and this technique for solving LPs which have an exponential number of variables is well-known in the literature (see \cite{Williamson,Vondrak_2011,Adamczyk2017,Lee2018} for instance).

We now define a \textit{fixed vertex} probing algorithm, called \textsc{VertexProbe}, which is
applied to an online vertex $s$ of an arbitrary stochastic graph (potentially distinct from $G$):
\begin{algorithm}[H]
\caption{VertexProbe}
\label{alg:vertex_probe}
\begin{algorithmic}[1]
\Require an online vertex $s$ of a stochastic graph, $\partial(s)$, and probabilities $(z(\bm{e}))_{\bm{e} \in \scr{C}_s}$ which satisfy $\sum_{\bm{e} \in \scr{C}_s} z(\bm{e}) = 1$. 
\Ensure an active edge $\scr{N}$ of $\partial(s)$.
\State $\scr{N} \leftarrow \emptyset$.
\State Draw $\bm{e}'$ from $\scr{C}_s$ with
probability $z(\bm{e}')$.
\If{$\bm{e}'= \lambda$}				\Comment{the empty string is drawn.}
\State \Return $\scr{N}$.
\Else{}
\State Denote $\bm{e}' = (e_{1}', \ldots ,e_{k}')$ for $k := |\bm{e}'| \ge 1$.
\For{$i=1, \ldots ,k$}
\State Probe the edge $e'_i$.
\If{$\sta(e'_i)=1$}
\State  Add $e'_i$ to $\scr{N}$, and exit the ``for loop''.
\EndIf
\EndFor
\EndIf
\State \Return $\scr{N}$.
\end{algorithmic}
\end{algorithm}

\begin{lemma}\label{lem:fixed_vertex_probe}
Suppose \textsc{VertexProbe} (Algorithm \ref{alg:vertex_probe}) is passed a fixed online node $s$ of a stochastic graph,
and probabilities $(z(\bm{e}))_{e \in \scr{C}_s}$ which satisfy $\sum_{\bm{e} \in \scr{C}_s} z(\bm{e}) =1$.
If for each $e \in \partial(s)$,
\[
\til{z}_{e}:= \sum_{\substack{\bm{e}' \in \scr{C}_v: \\ e \in \bm{e}'}} g(\bm{e}_{ < e}') \cdot z_{v}( \bm{e}'),
\] 
then $e \in \partial(s)$ is probed with probability $\til{z}_{e}$, and returned by the algorithm with probability $p_{e} \cdot \til{z}_e$.
\end{lemma}
\begin{remark}
If \textsc{VertexProbe} outputs the edge $e=(u,s)$ when executing on the fixed node $s$,
then we say that $s$ \textbf{commits} to the edge $e=(u,s)$,
or that $s$ commits to $u$.
\end{remark}
Returning to the problem of designing an online probing algorithm
for $G$, let us assume that $n:=|V|$, and that the online nodes of $V$ are denoted
$v_{1}, \ldots ,v_{n}$, where the order is generated $u.a.r.$.
Denote $V_{t}$ as the set of first $t$ arrivals of $V$; that is, $V_{t}:= \{v_{1}, \ldots ,v_{t}\}$. Moreover, set $G_{t}:= G[U \cup V_t]$, and $\cLPOPT(G_t)$ as the value of an optimum solution to \ref{LP:config} (this is a random variable, as $V_{t}$ is a random subset of $V$). The following inequality then holds:
\begin{lemma} \label{lem:random_induced_subgraph}
For each $t \ge 1$,  $\mb{E}[ \LPOPT_{conf}(G_{t}) ] \ge \frac{t}{n} \, \LPOPT_{conf}(G)$.
\end{lemma}
In light of this observation, we design an online probing algorithm which makes use of $V_{t}$, the currently known nodes, to derive an optimum LP solution with respect to $G_{t}$. As such, each time an online node
arrives, we must compute an optimum solution for the LP associated to $G_{t}$, distinct from the solution computed
for that of $G_{t-1}$. 

\begin{algorithm}[H]
\caption{Unknown Stochastic Graph ROM} 
\label{alg:ROM_edge_weights}
\begin{algorithmic}[1]
\Require $U$ and $n:=|V|$.
\Ensure a matching $\scr{M}$ from the (unknown) stochastic graph $G=(U,V,E)$ of active edges.
\State Set $\scr{M} \leftarrow \emptyset$.
\State Set $G_{0} = (U, \emptyset, \emptyset)$
\For{$t=1, \ldots , n$}
\State Input $v_{t}$, with $(w_{e})_{e \in \partial(v_t)}$, $(p_{e})_{e \in \partial(v_t)}$ and $\scr{C}_{v_t}$.  
\State Compute $G_{t}$, by updating $G_{t-1}$ to contain $v_{t}$ (and its relevant information).
\If{ $t < \lfloor n/e \rfloor$}
\State Pass on $v_{t}$.
\Else
\State Solve \ref{LP:config} for $G_{t}$ and find an optimum solution $(x_{v}(\bm{e}))_{v \in V_{t}, \bm{e} \in \scr{C}_v}$. \label{line:cur_solution}
\State Set $e_t \leftarrow \textsc{VertexProbe}(v_t, \partial(v_t), (x_{v}(\bm{e}))_{\bm{e} \in \scr{C}_{v_t}})$.
\If{$e_t=(u_t,v_t) \neq \emptyset$ and $u_t$ is unmatched}
\State Add $e_t$ to $\scr{M}$.
\EndIf
\EndIf
\EndFor
\State \Return $\scr{M}$.
\end{algorithmic}
\end{algorithm}

Let us consider the matching $\scr{M}$ returned by the algorithm,
as well as its weight, which we denote by $w(\scr{M})$. Set $\alpha:=1/e$ for clarity, and take $t \ge \lceil \alpha n \rceil$. Define $R_{t}$ as the \textit{unmatched vertices} of $U$ when vertex $v_{t}$ arrives. Note that committing to $e_t=(u_t,v_t)$ is necessary, but not sufficient, for $v_t$ to match to $u_t$.
With this notation, we have that $\mb{E}[ w(\scr{M}) ] = \sum_{t=\lceil \alpha n \rceil}^{n} \mb{E}[ w(u_t, v_t) \cdot \bm{1}_{[u_t \in R_t]} ]$.
Moreover, we claim the following:
\begin{lemma} \label{lem:edge_value_lower_bound}
For each  $t \ge \lceil \alpha n \rceil$, $\mb{E}[ w(e_t)] \ge \LPOPT_{conf}(G)/n$.
\end{lemma}
\begin{proof}[Proof of Lemma \ref{lem:edge_value_lower_bound}]
Set $\alpha:=1/e$ for clarity, and take $t \ge \lceil \alpha n \rceil$. 
Define $e_{t}:=(u_t,v_{t})$, where $u_{t}$ is the vertex of $U$ which $v_{t}$ commits to
(which is the empty set $\emptyset$, if no such commitment occurs). 
For each $u \in U$, denote $C(u,v_t)$ as the event in which
$v_t$ commits to $u$. 
Let us now condition on the random subset $V_{t}$, as well as the 
random vertex $v_{t}$. 
In this case, 
\[
    \mb{E}[ w(e_{t})  \, | \, V_{t}, v_t] = \sum_{u \in U} w_{u,v_t} \, \mb{P}[ C(u,v_t)  \, | \, V_{t}, v_{t}].
\]
Observe however that once we condition on $V_{t}$ and $v_{t}$, Algorithm \ref{alg:ROM_edge_weights} corresponds to executing \\\textsc{VertexProbe} on the instance $(v_t, \partial(v_t), (x_{v_t}(\bm{e}))_{\bm{e} \in \scr{C}_v})$,
where we recall that $(x_{v}(\bm{e}))_{\bm{e} \in \scr{C}_v, v \in v_t}$ is an optimum solution
to \ref{LP:config} for $G_t=G[U \cup V_t]$. Thus,
Lemma \ref{lem:fixed_vertex_probe} implies that $\mb{P}[ C(u,v_t)  \, | \, V_{t}, v_{t}]  =  p_{u,v_t} \til{x}_{u,v_t}$,
where
\[
\til{x}_{u,v_t}:= \sum_{\substack{\bm{e}' \in \scr{C}_{v_t}: \\ e \in \bm{e}'}} g(\bm{e}_{ < e}') \cdot x_{v_t}( \bm{e}'),
\]
and so $\mb{E}[ w(e_{t})  \, | \, V_{t}, v_t] = \sum_{u \in U} w_{u,v_t}  p_{u,v_t}  \til{x}_{u,v_t}$.
On the other hand, if we condition on \textit{solely} $V_{t}$, then $v_{t}$ remains distributed uniformly
at random amongst the vertices of $V_{t}$. Moreover, once we condition on $V_t$, the graph $G_t$ is determined, and thus so are
the values $(x_{v}(\bm{e}))_{v \in V_{t}, \bm{e} \in \scr{C}_v}$. These observations together imply that
\begin{equation}\label{eqn:conditional_expectation_value}
\mb{E}[ w_{u,v_t} \, p_{u,v_t} \,  \til{x}_{u,v_t}  \, | \, V_{t}] = \frac{\sum_{v \in V_t} w_{u,v} \, p_{u,v} \, \til{x}_{u,v}}{t}
\end{equation}
for each $u \in U$ and $\lceil \alpha n \rceil \le t \le n$. If we now take expectation over $v_{t}$, then using the law of iterated
expectations,
\begin{align*}
    \mb{E}[ w(e_t)  \, | \, V_{t}] &= \mb{E}[ \,  \mb{E}[ w(e_t) \, | \, V_t, v_t] \,   \, | \, V_{t} ]  \\
                               &= \mb{E}\left[ \sum_{u \in U} w_{u,v_t} \, p_{u,v_t} \, \til{x}_{u,v_t}   \, | \, V_{t} \right] \\
                               &= \sum_{u \in U} \mb{E}[ w_{u,v_t} \, p_{u,v_t} \, \til{x}_{u,v_t}  \, | \, V_{t}] \\
                               &= \sum_{u \in U} \sum_{v \in V_t} \frac{w_{u,v} p_{u,v} \, \til{x}_{u,v}}{t},
\end{align*}
where the final equation follows from \eqref{eqn:conditional_expectation_value}.
Observe however that
\[
    \cLPOPT(G_t)=\sum_{v \in V_t} \sum_{u \in U} w_{u,v_t} \, p_{u,v_t} \, \til{x}_{u,v_t},
\]
as $(x_{v}(\bm{e}))_{v \in V_{t}, \bm{e} \in \scr{C}_v}$ is an optimum solution to
\ref{LP:config} for $G_t$. As a result,
\[
    \mb{E}[ w(e_t)  \, | \, V_{t}]  = \frac{\cLPOPT(G_t)}{t},
\]
and so
\[
	\mb{E}[ w(e_t)]  = \frac{\mb{E}[\cLPOPT(G_t)]}{t},
\]
after taking taking expectation over $V_{t}$. On the other hand, Lemma \ref{lem:random_induced_subgraph} implies that
\[
	\frac{\mb{E}[\cLPOPT(G_t)]}{t} \ge  \frac{\cLPOPT(G)}{n}.
\]
Thus,
\[
\mb{E}[ w(e_t)] \ge \frac{\cLPOPT(G)}{n},
\]
provided $\lceil \alpha n \rceil \le t \le n$, thereby completing the proof.

\end{proof}

\begin{lemma} \label{lem:availability_lower_bound}
For each $t \ge \lceil \alpha n \rceil$, define $f(t,n):= \lfloor \alpha n \rfloor /(t-1)$. In this case, $\mb{P}[ u_{t} \in R_t  \, | \, V_{t},v_t] \ge f(t,n)$,
where $V_{t}=\{v_{1},\ldots ,v_{t}\}$ and $v_t$ is the $t^{th}$ arriving node of $V$ \footnote{Note that since
$V_t$ is a set, conditioning on $V_t$ only reveals which vertices of $V$ encompass the first $t$ arrivals,
\textit{not} the order they arrived in. Hence, conditioning on $v_t$ as well reveals strictly more information.}.
\end{lemma}

\begin{proof}[Proof of Lemma \ref{lem:availability_lower_bound}]

Let us assume that $\lceil \alpha n  \rceil \le t \le n$ is fixed, and $(x_{v}^{(t)}(\bm{e}))_{v \in V, \bm{e} \in \scr{C}_v}$
is the optimum solution of \ref{LP:config} for $G_t$, as used by Algorithm \ref{alg:ROM_edge_weights}.
For each $u \in U$ and $v \in v$, define the \textbf{edge variable} $\til{x}^{(t)}_{u,v}$,
where
\[
	\til{x}^{(t)}_{u,v} := \sum_{\substack{\bm{e}' \in \scr{C}_{v_t}: \\ e \in \bm{e}'}} g(\bm{e}_{ < e}') \cdot x_{v_t}^{(t)}( \bm{e}')
\]
We wish to prove that for each $u \in U$, 
\begin{equation} \label{eqn:conditional_available}
\mb{P}[u \in R_{t}  \, | \, V_{t},v_t] \ge \lfloor \alpha n \rfloor /(t-1).
\end{equation}
As such, we must condition on $(V_t,v_t)$ throughout the remainder of the proof.
To simplify the argument, we abuse notation slightly and remove $(V_t,v_t)$
from the subsequent probability computations, though it is understood to implicitly
appear.

Given arriving node $v_{j}$ for $j=1, \ldots ,n$, once again denote $C(u,v_j)$
as the event in which $v_j$ commits to $u \in U$. As $R_{t}$ denotes the
unmatched nodes after the vertices $v_{1}, \ldots , v_{t-1}$ are processed
by Algorithm \ref{alg:ROM_edge_weights}, observe that
$u \in R_{t}$ if and only if $\neg C(u,v_j)$ occurs for each $j=1, \ldots , t-1$,
and so $\mb{P}[ u \in R_{t} ] = \mb{P}[\cap_{j=1}^{t-1} \neg C(u,v_j)]$.
We therefore focus on lower bounding $\mb{P}[\cap_{j=1}^{t-1} \neg C(u,v_j) ]$ in order to prove
the lemma.

First observe that for $j=1, \ldots , \lfloor \alpha n \rfloor$, the algorithm passes on probing
$\partial(v_j)$ by definition, and so \eqref{eqn:conditional_available} holds
if $t = \lceil \alpha n \rceil$. As such, we may thereby
assume $t \ge \lceil \alpha n \rceil +1$ and focus on lower bounding $\mb{P}[\cap_{j= \lceil \alpha n \rceil}^{t-1} \neg C(u,v_j)]$. Observe that this event depends only on the probes of the vertices of $V_{t-1} \setminus V_{\lfloor \alpha n \rfloor}$.
We denote $\bar{t}:= t - 1 - \lfloor \alpha n \rfloor = t - \lceil \alpha n \rceil$ as the number of vertices within this set.

Let us first consider the vertex $v_{t-1}$,
and the edge variable $\til{x}^{(t-1)}_{u,v}$ for each $v \in V_{t-1}$.
Observe that after applying Lemma \ref{lem:fixed_vertex_probe},
\begin{align*}
	\mb{P}[ C(u,v_{t-1})] &= \sum_{ v \in V_{t-1}} \mb{P}[ C(u,v_{t-1})  \, | \, v_{t-1} = v] \cdot \mb{P}[v_{t-1}=v] \\
											  &= \frac{1}{t-1} \sum_{v \in V_{t-1}} \til{x}_{u,v}^{(t-1)}  p_{u,v},
\end{align*}
as once we condition on the set $V_{t}$
and the vertex $v_t$, we know that $v_{t-1}$ is uniformly distributed amongst $V_{t-1}$.
On the other hand, the values $(\til{x}_{u,v}^{(t-1)})_{u \in U, v \in V_{t-1}}$ are derived
from a solution to \ref{LP:config} for $G_{t-1}$, and so by constraint \eqref{eqn:relaxation_efficiency_matching},
\[
	\sum_{v \in V_{t-1}} \til{x}_{u,v}^{(t-1)}  p_{u,v} \le 1.
\]
We therefore get that $\mb{P}[ C(u,v_{t-1}) ]  \le \frac{1}{t-1}$.
Similarly, if we fix $1 \le k \le \bar{t}$, then we can generalize the above
argument by conditioning
on the identities of all the vertices preceding $v_{t-k}$, as well as the probes
they make; that is, $(u_{t-1},v_{t-1}), \ldots ,(u_{t-(k-1)},v_{t-(k-1)})$ (in addition to $V_t$ and $v_t$ as always). 

In order to simplify the resulting indices, let us reorder the vertices of $V_{t-1} \setminus V_{\lfloor \alpha n \rfloor }$.
Specifically, define $\bar{v}_{k}:=v_{t-k}, \bar{u}_k := u_{t-k}$ and $\bar{e}_k:=e_{t-k}$
for $k=1,\ldots ,\bar{t}$. With this notation, we denote
$\scr{H}_{k}$ as encoding the information available based on the vertices $\bar{v}_1, \ldots , \bar{v}_{k}$ and the edges they (potentially) committed to, namely $\bar{e}_{1}, \ldots ,\bar{e}_{k}$ \footnote{Formally, $\scr{H}_k$ is the sigma-algebra
generated from $V_t,v_t$ and $\bar{e}_{1}, \ldots ,\bar{e}_{k}$.}. By convention,
we define $\scr{H}_{0}$ as encoding the information regarding $V_t$ and $v_t$. 

An analogous computation to the above case then implies that
\[
	\mb{P}[ C(u,\bar{v}_{k})  \, | \, \scr{H}_{k-1}] = \sum_{ v \in V_{t-k}} \til{x}_{u,v}^{(t-k)}  p_{u,v}  \mb{P}[\bar{v}_k=v] \le \frac{1}{t-k},
\]
for each $k=1,\ldots , \bar{t}$, where
$\til{x}_{u,v}^{(t-k)}$ is the edge variable for $v \in V_{t-k}$.

Observe now that in each step, we condition on strictly more information;
that is, $\scr{H}_{k-1} \subseteq \scr{H}_{k}$
for each $k=2, \ldots , \bar{t}$. On the other hand, observe that
if we condition on $\scr{H}_{k-1}$ for $1 \le k \le \bar{t}-1$,
then the event $C(u,\bar{v}_{j})$ can be determined from $\scr{H}_{k-1}$
for each $1 \le j \le k-1$.

Using these observations, for $1 \le k \le \bar{t}$, the following recursion holds:
\begin{align*}
\mb{P}[ \cap_{j=1}^{k} \neg C(u,\bar{v}_j)] 
  &= \mb{E}\left[ \, \mb{E}\left[ \prod_{j=1}^{k} \bm{1}_{ [\neg C(u,\bar{v}_j)]}  \, | \, \scr{H}_{k-1}\right] \right]	\\
&= \mb{E}\left[ \,  \prod_{j=1}^{k-1} \bm{1}_{ [\neg C(u,\bar{v}_j)]} \, \mb{P}[\neg C(u, \bar{v}_k)  \, | \, \scr{H}_{k-1}]\right] \\
&\ge \left(1 - \frac{1}{t-k}\right) \mb{P}[ \cap_{j=1}^{k-1} \neg C(u,\bar{v}_j)]
\end{align*}
It follows that if $k= \bar{t}= t - \lceil \alpha n \rceil$, then applying the above recursion implies that
\[
	\mb{P}[\cap_{j= \lceil \alpha n \rceil}^{t-1} \neg C(u,v_j)] \ge \prod_{k=1}^{t- \lceil \alpha n \rceil} \left(1 - \frac{1}{t-k}\right).
\]
Thus, after cancelling the pairwise products,
\[
	 \mb{P}[ u \in R_{t}] = \mb{P}[\cap_{j= \alpha n}^{t-1}\neg C(u,v_j)] \ge  \frac{\lfloor \alpha n \rfloor}{t-1},
\]
and so \eqref{eqn:conditional_available} holds for all $t \ge \lceil \alpha n \rceil$, thereby completing the argument.

\end{proof}

With these lemmas, together with the efficient solvability of \ref{LP:config},
the proof of Theorem \ref{thm:ROM_edge_weights} follows easily:

\begin{proof}[Proof of Theorem \ref{thm:ROM_edge_weights}]
Clearly, Algorithm \ref{alg:ROM_edge_weights} can be implemented efficiently, since \ref{LP:config} is efficiently solvable,
provided it involves a stochastic graph whose probing constraints are permutation-closed.
Thus, we focus on proving the algorithm attains the desired asymptotic competitive ratio.

Let us consider the matching $\scr{M}$ returned by the algorithm,
as well as its weight, which we denote by $w(\scr{M})$. Set $\alpha:=1/e$ for clarity, and take $t \ge \lceil \alpha n \rceil$,
where we define $R_{t}$ to be the \textit{unmatched vertices} of $U$ when vertex $v_{t}$ arrives. Moreover, define $e_{t}$
as the edge $v_{t}$ commits to, which is the empty-set by definition if no such commitment
is made. Observe that
\begin{equation}\label{eqn:value_of_matching}
   \mb{E}[ w(\scr{M}) ] = \sum_{t=\lceil \alpha n \rceil}^{n} \mb{E}[ w(u_t, v_t) \cdot \bm{1}_{[u_t \in R_t]} ].
\end{equation}
Fix $\lceil \alpha n \rceil \le t \le n$, and first observe that $w(u_t, v_t)$ and $\{ u_{t} \in R_t \}$ are conditionally independent given
$(V_{t},v_t)$, as the probes involving $\partial(v_t)$ are independent from those of $v_{1}, \ldots ,v_{t-1}$. Thus,
\[
\mb{E}[ w(u_t,v_t) \cdot \bm{1}_{[u_t \in R_t]}  \, | \, V_{t}, v_{t}] = \mb{E}[ w(u_t,v_t)  \, | \, V_t, v_t] \cdot \mb{P}[ u_t \in R_t  \, | \, V_{t}, v_t].
\]
Moreover, Lemma \ref{lem:availability_lower_bound} implies that
\[    
   \mb{E}[ w(u_t,v_t)  \, | \, V_t, v_t] \cdot \mb{P}[ u_t \in R_t  \, | \, V_{t}, v_t] \ge 
   \mb{E}[ w(u_t,v_t)  \, | \, V_{t}, v_t] f(t,n),
\]
and so $\mb{E}[ w(u_t,v_t) \, \bm{1}_{[u_t \in R_t]}  \, | \, V_{t}, v_{t}] \ge \mb{E}[ w(u_t,v_t)  \, | \, V_{t}, v_t] \, f(t,n)$.
Thus, by the law of iterated expectations\footnote{$\mb{E}[ w(u_t,v_t) \cdot \bm{1}_{[u_t \in R_t]}  \, | \, V_{t}, v_{t}]$
is a random variable which depends on $V_t$ and $v_t$, and so the outer expectation is over the randomness in $V_t$ and $v_t$.}
\begin{align*}
\mb{E}[ w(u_t,v_t) \cdot \bm{1}_{[u_t \in R_t]} ] &= \mb{E}[ \, \mb{E}[ w(u_t,v_t) \cdot \bm{1}_{[u_t \in R_t]}  \, | \, V_{t}, v_{t}] \, ] \\
			&\ge \mb{E}[ \, \mb{E}[ w(u_t,v_t)  \, | \, V_{t}, v_t]  f(t,n) \, ] 
			= f(t,n) \mb{E}[ w(u_t,v_t)].
\end{align*}
As a result, using \eqref{eqn:value_of_matching}, we get that
\begin{align*}
	\mb{E}[w(\scr{M})] &= \sum_{t=\lceil \alpha n \rceil}^{n} \mb{E}[ w(u_t, v_t) \, \bm{1}_{[u_t \in R_t]} ]
						  \ge \sum_{t=\lceil \alpha n \rceil}^{n} f(t,n) \, \mb{E}[ w(u_t,v_t)].
\end{align*}
We may thus conclude that
\[
    \mb{E}[ w(\scr{M})] \ge \LPOPT_{conf}(G) \sum_{t=\lceil \alpha n \rceil}^{n} \frac{ f(t,n)}{n},
\]
after applying Lemma \ref{lem:edge_value_lower_bound}. 
As $\sum_{t=\lceil \alpha n \rceil}^{n} f(t,n)/n \ge (1/e -1/n)$, the result holds.

\end{proof}

\section{Vertex Weights} \label{sec:vertex_weights}



Suppose that $G=(U,V,E)$ is a vertex weighted stochastic graph
with weights $(w_u)_{u \in U}$.
Note that it will be convenient to denote $w_{u,v}:= w_{u}$ provided $(u,v) \in \partial(v)$ for $v \in V$.
Let us now fix $s \in V$, and recall that $\val(\bm{e})$
is the expected weight of the edge matched, provided the edges
of $\bm{e}$ are probed in order, where $\bm{e} \in \scr{C}_{s}$.
Observe then the following claim, which builds upon the work of Brubach et al. \cite{Brubach2019}, and
before that, Purohit et al. \cite{Purohit2019}:
\begin{theorem}\label{thm:efficient_star_dp}
There exists an dynamical programming (DP) based algorithm \textsc{DP-OPT}, which
given access to $G[ \{s\} \cup U]$, computes a tuple $\bm{e}' \in \scr{C}_{s}$, such
that $\OPT(s,U) = \val(\bm{e}')$. Moreover, if $\scr{C}_s$ is closed
under substrings and permutations, then \textsc{DP-OPT}
is efficient, assuming access to a membership oracle for $\scr{C}_s$.
\end{theorem}
\begin{proof}[Proof of Theorem \ref{thm:efficient_star_dp}]
It will also be convenient to denote $w_{u,s}:=w_{u}$ for each $u \in U$
such that $(u,s) \in \partial(s)$.

We first must show that there exists some $\bm{e}' \in \scr{C}_{s}$
such that $\val(\bm{e}') = \OPT(s,U)$, where 
\begin{equation}\label{eqn:expected_value_of_edge_probes}
	\val(\bm{e}):= \sum_{i=1}^{|\bm{e}|} p_{e_i} w_{e_i} \prod_{j=1}^{i-1} (1 - p_{e_i}),
\end{equation}
for $\bm{e} \in \scr{C}_s$, and $\OPT(s,U)$ is the value of the committal benchmark on $G[\{s\}\cup U]$.
Since the committal benchmark must respect commitment -- i.e., match the first
edge to $s$ which it reveals to be active -- it is clear that $\bm{e}'$ exists.

Let us now additionally assume that $\scr{C}_s$ is also closed under permutations.
Our goal is to show that $\bm{e}'$ can be computed efficiently.
Now, for any $\bm{e} \in \scr{C}_s$, let $\bm{e}^{r}$ be the rearrangement of $\bm{e}$, based on the non-increasing order
of the weights $(w_{e})_{e \in \bm{e}}$. Since $\scr{C}_s$ is closed under permutations,
we know that $\bm{e}^{r}$ is also in $\scr{C}_s$. Moreover, $\val(\bm{e}^{r}) \ge \val(\bm{e})$.
Hence, let us order the edges of $\partial(s)$ as $e_{1}, \ldots ,e_{m}$,
such that $w_{e_1} \ge \ldots \ge  w_{e_m}$, where $m:=|\partial(s)|$.
Observe then that it suffices to maximize \eqref{eqn:expected_value_of_edge_probes} over
those strings within $\scr{C}_s$ which respect this ordering on $\partial(s)$.
Stated differently, let us denote $\scr{I}_{s}$ as the family of subsets of $\partial(s)$
induced by $\scr{C}_s$, and define the set function $f: 2^{\partial(s)} \rightarrow [0, \infty)$,
where $f(B):= \val(\bm{b})$ for $B=\{b_{1}, \ldots ,b_{|B|}\} \subseteq \partial(s)$,
such that $\bm{b}=(b_{1}, \ldots ,b_{|B|})$ and $w_{b_1} \ge \ldots \ge w_{b_{|B|}}$.
Our goal is then to efficiently maximize $f$ over the set-system $(\partial(s),\scr{I}_s)$.
Observe that since $\scr{C}_s$ is both substring-closed and permutation-closed, 
$\scr{I}_s$ is downward closed. Moreover, clearly we can simulate oracle access to
$\scr{I}_s$, based on our oracle access to $\scr{C}_s$.

For each $i=0, \ldots ,m-1$, denote $\partial(s)^{>i}:=\{e_{i+1}, \ldots ,e_{m}\}$,
and $\partial(s)^{>m}:= \emptyset$. Moreover, define the family of subsets $\scr{I}_{s}^{>i}:= \{B \subseteq \partial(s)^{>i} : B \cup \{e_i\} \in \scr{I}_s\}$ for each $2 \le i \le m$, and $\scr{I}_{s}^{>0}:= \scr{I}_s$. Observe then that
$(\partial(s)^{>i}, \scr{I}_{s}^{>i})$ is a downward-closed set system, as $\scr{I}_s$ is downward-closed.
Moreover, we may simulate oracle access to $\scr{I}^{>i}_{s}$ based on our oracle access to $\scr{I}_s$.

Denote $\OPT(\scr{I}_{s}^{>i})$ as the maximum value of $f$ over constraints $\scr{I}_{s}^{>i}$.
Observe then the following recursion:
\begin{equation} \label{eqn:dynamical_program}
	\OPT(\scr{I}_{s}) :=  
			\max_{i \in [m]} ( p_{e_i} \cdot w_{e_i} + (1 - p_{e_i}) \cdot \OPT(\scr{I}_{s}^{>i}) )
\end{equation}
Hence, given access to the values $\OPT(\scr{I}_{s}^{>1}), \ldots , \OPT(\scr{I}_{s}^{>m})$,
we can compute $\OPT(\scr{I}_s)$ efficiently. In fact, it is clear that we can use \eqref{eqn:dynamical_program}
to recover an optimum solution to $f$, and so the proof follows by an inductive argument on $|\partial(s)|$.
We can define \textsc{DP-OPT} to be a memoization based implementation of \eqref{eqn:dynamical_program}.
\end{proof}

Given $R \subseteq U$, consider the induced stochastic graph,
$G[\{s\} \cup R]$ for $R \subseteq U$ which has probing constraint $\scr{C}_{s}^{R} \subseteq \scr{C}_v$,
constructed by restricting $\scr{C}_s$ to those strings whose
entries all lie in $R \times \{s\}$. Moreover, denote
the output of executing \textsc{DP-OPT} on $G[\{s\} \cup R]$ by $\textsc{DP-OPT}(s,R)$. Consider now the following online probing
algorithm, where we assume the online vertices of $G$ arrive in an adversarially chosen unknown order $v_{1}, \ldots ,v_{n}$,
where $n:=|V|$.
\begin{algorithm} 
\caption{Greedy-DP}\label{alg:dynamical_program}
\begin{algorithmic}[1] 
\Require offline vertices $U$ with vertex weights $(w_{u})_{u \in U}$.
\Ensure a matching $\scr{M}$ of active edges of the unknown stochastic graph $G=(U,V,E)$.
\State $\scr{M} \leftarrow \emptyset$.
\State $R \leftarrow U$.
\For{$t=1, \ldots , n$}
\State Let $v_{t}$ be the current online arrival node, with constraint $\scr{C}_{v_t}$ and edges probabilities $(p_{e})_{e \in \partial(v_t)}$.
\State Set $\bm{e}  \leftarrow \textsc{DP-OPT}(v_t,R)$
\For{$i=1, \ldots , |\bm{e}|$}
\State Probe $e_i$.
\If{$\sta(e_i) =1$}
\State Add $e_i$ to $\scr{M}$, and update $R \leftarrow R \setminus \{u_i\}$,
where $e_{i}=(u_i,v_t)$.
\EndIf
\EndFor
\EndFor
\State \Return $\scr{M}$.
\end{algorithmic}
\end{algorithm}

In general, the behaviour of $\textsc{OPT}(s,R)$ can change very much, even for minor changes to $R$.
For instance, if $R= U$, then $\textsc{OPT}(s,U)$ may probe $(u,s)$ first -- thus giving it highest priority -- whereas by removing $u^* \in U$ from $U$ (where $u^* \neq u$),
$\textsc{OPT}(s, U \setminus \{u^*\})$ may not probe $(u,v)$ at all:
\begin{example} \label{example:bad_OPT_behaviour}
Let $G = (U, V, E)$ be a bipartite graph with $U = \{u_1, u_2, u_3, u_4\}$, $V = \{v\}$ and $\ell_v =2$. 
Set $p_{u_1, v} = 1/3$, $p_{u_2, v} = 1$, $p_{u_3, v} = 1/2$, $p_{u_4, v} = 2/3$.
Fix $\eps > 0$, and let the weights of offline vertices be $w_{u_1} = 1 + \eps$, $w_{u_2} = 1 + \eps/2$, $w_{u_3} = w_{u_4} = 1$. 
We assume that $\eps$ is sufficiently small -- concretely, $\eps \le 1/12$.
If $R_{1}:= U$, then $\OPT(v,R_1)$ probes $(u_{1},v)$ and then $(u_{2},v)$ in order. On the other hand,
if $R_{2} = R_{1} \setminus \{v_{2}\}$, then $\OPT(v,R_{2})$ does \textit{not} probe $(u_{1},v)$. Specifically,
$\OPT(v,R_{2})$ probes $(u_3,v)$ and then $(u_4,v)$.
\end{example}
While this behaviour isn't problematic in the case of adversarial
arrivals, we must restrict our attention to executions of Algorithm \ref{alg:dynamical_program}
which are less \textit{adaptive} for our primal-dual proof to work in the case of ROM arrivals. 

Given a vertex $v \in V$, and an ordering $\pi_v$ on $\partial(v)$,
if $R \subseteq U$, then define $\pi_{v}(R)$ to be the longest\footnote{Given $\bm{e}'$ after processing $e_1,\ldots ,e_i$
via ordering $\pi_v$, append $e_{i+1}$ if $(\bm{e}',e_{i+1}) \in \scr{C}_{v}^{R}$, else move to $e_{i+2}$.} string constructible by iteratively appending the edges of $R \times \{v\}$ via $\pi_v$, subject to respecting constraint $\scr{C}^{R}_v$. We say that $v$ is \textbf{rankable}, provided there exists a choice of $\pi_{v}$ which depends \textit{solely} on $(p_e)_{e \in \partial(v)}$, $(w_{e})_{e \in \partial(v)}$ and $\scr{C}_v$, such that for \textit{every} $R \subseteq U$, the strings $\textsc{DP-OPT}(v,R)$ and $\pi_{v}(R)$
are equal. Crucially, if $v$ is rankable, then when vertex $v$ arrives while executing Algorithm \ref{alg:dynamical_program}, one can compute the ranking $\pi_{v}$ on $\partial(v)$ and probe the adjacent edges of $R \times \{v\}$  based on this order, subject to not violating the constraint $\scr{C}_{v}^{R}$. By following this probing strategy, the optimality of \textsc{DP-OPT} ensures that the expected weight of the match made to $v$ will be $\OPT(v,R)$.
We consider three (non-exhaustive) examples of rankability:
\begin{proposition} \label{prop:rankability}
Let $G=(U,V,E)$ be a stochastic graph, and suppose that $v \in V$.
If $v$ satisfies either of the following conditions, then $v$ is rankable:
\begin{enumerate}
\item $v$ has unit patience or unlimited patience; that is, $\ell_v \in \{1, |U|\}$. \label{eqn:rankability_examples_patience}
\item $v$ has patience $\ell_v$, and for each $u_{1}, u_{2} \in U$, if $p_{u_{1},v} \le p_{u_{2},v}$ then $w_{u_{1}} \le w_{u_{2}}$.
\item $G$ is unweighted, and $v$ has a budget\footnote{In the case of a budget $B_v$ and edge probing costs $(c_v)_{v \in V}$, 
any subset of $\partial(v)$ may be probed, provided its cumulative cost does not exceed $B_v$.} $B_v$ with edge probing costs $(c_{u,v})_{u \in U}$, and for each $u_1, u_2 \in U$, if $p_{u_1,v} \le p_{u_{2},v}$ then $c_{u_1,v} \ge c_{u_{2},v}$.
\end{enumerate}
\end{proposition}


We refer to the stochastic graph $G$ as \textbf{rankable}, provided all of its
vertices are themselves rankable. We emphasize that distinct vertices of $V$
may each use their own separate rankings of their adjacent edges. 

\begin{theorem}\label{thm:ROM_rankable}
Suppose Algorithm \ref{alg:dynamical_program} returns the matching $\scr{M}$ when
executing on the rankable stochastic graph $G=(U,V,E)$ with prefix-closed constraints $(\scr{C}_v)_{v \in V}$. 
In this case,
\[
\mb{E}[ w(\scr{M}) ] \ge  \left(1 - \frac{1}{e} \right) \cdot \OPT(G),
\]
provided the vertices of $V$ arrive $u.a.r.$. Algorithm \ref{alg:dynamical_program}
can be implemented efficiently, provided the constraints $(\scr{C}_v)_{v \in V}$ are also permutation-closed.
\end{theorem}
\begin{remark}
Proposition \ref{prop:rankability} includes the unit patience setting of Theorem \ref{thm:ROM_unit_patience}, 
as well as when $G$ is unweighted and has arbitrary patience values, and so we focus on proving Theorem \ref{thm:ROM_rankable} for the remainder of the section.
\end{remark}

In order to show that Algorithm \ref{alg:dynamical_program} attains
the claimed competitive ratios, we upper bound $\OPT(G)$
using an LP relaxation which accounts for arbitrary probing constraints.
For each $u \in U$ and $v \in V$, let $x_{u,v}$ be a decision variable
corresponding to the probability that $\OPT(G)$ probes the edge $(u,v)$.
\begin{align}\label{LP:DP}
\tag{LP-DP}
& \text{maximize}  & \sum_{u \in U} \sum_{v \in V} w_u \cdot p_{u, v}  \cdot x_{u, v}\\
&\text{subject to} &\, \sum_{v \in V}  p_{u, v} \cdot x_{u, v} &\leq 1 &&\forall u \in U       \label{eqn:dp_matching_constraint} \\
&  &\sum_{u \in R} w_u \cdot p_{u,v} \cdot x_{u,v} & \le  \OPT(v,R)  && \forall v \in V, \, R \subseteq U	\label{eqn:OPT_constraint}\\ 
&  &x_{u, v} & \ge 0 && \forall u \in U, v \in V		\label{eqn:edge_upper_bound}
\end{align}
Denote $\dLPOPT(G)$ as the optimum value of this LP. 
Constraint \eqref{eqn:dp_matching_constraint} can be viewed as ensuring that
the expected number of matches made to $u \in U$ is at most $1$. Similarly,
\eqref{eqn:OPT_constraint} can be interpreted as ensuring
that expected stochastic reward of $v$, suggested by the solution $(x_{u,v})_{u \in U,v \in V}$,
is actually attainable by the committal benchmark. 
Thus, $\OPT(G) \le \dLPOPT(G)$ (a formal proof specific to patience
values is proven in \cite{Brubach2019}).


\subsubsection{Defining the Primal-Dual Charging Schemes} \label{sec:charging_scheme}

In order to prove Theorems \ref{thm:adversarial} and \ref{thm:ROM_rankable}, we employ primal-dual charging arguments
based on the dual of \ref{LP:DP}. For each $u \in U$, define the variable $\alpha_u$. Moreover, for each $R \subseteq U$
and $v \in V$, define the variable $\phi_{v,R}$ (these latter variables correspond to constraint \eqref{eqn:OPT_constraint}). 
\begin{align} \label{LP:DP_dual}
\tag{LP-dual-DP}
&\text{minimize}&   \sum_{u \in U} \alpha_u + \sum_{v \in V}\sum_{R \subseteq U}\OPT(v,R) \cdot \phi_{v,R}	\\
&\text{subject to}&  \: p_{u, v} \cdot \alpha_u  + \sum_{\substack{R \subseteq U: \\ u \in R}} w_u \cdot p_{u,v} \cdot \phi_{v,R} &\geq  w_u \cdot p_{u, v} && \forall u \in U, v \in V\\
& &   \alpha_u &\geq 0 && \forall u \in U\\
& &     \phi_{v,R} & \ge 0 && \forall v \in V, R \subseteq U
\end{align}
The dual-fitting argument used to prove Theorem \ref{thm:ROM_rankable} has
an initial set-up which proceeds as in Devanur et al. \cite{DJK2013}.
Specifically, let $F:= 1 -1/e$ and define $g: [0,1] \rightarrow [0,1]$ 
where $g(z):= \exp(z-1)$ for $z \in [0,1]$. For each $v \in V$, draw $Y_{v} \in [0,1]$ independently and uniformly at random. We assume that the vertices of $V$ are presented to Algorithm \ref{alg:dynamical_program} in a non-decreasing order, based on the values of $(Y_{v})_{v \in V}$.

We now describe how the charging assignments are made while Algorithm \ref{alg:dynamical_program} executes
on $G$.  Firstly, we initialize a dual solution $((\alpha_u)_{u \in U}, (\phi_{v,R})_{v \in V, R \subseteq U})$ where all the variables are set equal to $0$. Let us now take $v \in V, u \in U$, and $R \subseteq U$, where $u \in R$.
If $R$ consists of the unmatched vertices of $v$ when it arrives at time $Y_{v}$,
then suppose that Algorithm \ref{alg:dynamical_program} matches $v$ to $u$ while making its probes to
a subset of the edges of $R \times \{v\}$. In this case, we \textbf{charge} $w_{u} \cdot (1 - g(Y_{v}))/F$ to $\alpha_{u}$ and
$w_{u} \cdot g(Y_{v})/ (F \cdot \OPT(v,R))$ to $\phi_{v,R}$.
Observe that each subset $R \subseteq U$ is charged at most once, as is each $u \in U$.
Thus, by definition,
\begin{equation}\label{eqn:expected_competitive_ratio_dual_general}
	\mb{E}[ w( \scr{M})] = F \cdot \left( \sum_{u \in U} \mb{E}[\alpha_{u}]  +  \sum_{v \in V} \sum_{R \subseteq U} \OPT(v,R) \cdot \mb{E}[ \phi_{v,R}] \right),
\end{equation}
where the expectation is over the random variables $(Y_{v})_{v \in V}$
and $(\sta(e))_{e \in E}$. If we now set $\alpha^{*}_{u} := \mb{E}[\alpha_{u} ]$ and $\phi^{*}_{v,R} := \mb{E}[ \phi_{v,R}]$ for $u \in U, v \in V$ and $R \subseteq U$, then \eqref{eqn:expected_competitive_ratio_dual_general} implies the following lemma:


%

\begin{lemma} \label{lem:dual_variables_expected_value}
Suppose $G=(U,V,E)$ is a stochastic graph for which
Algorithm \ref{alg:dynamical_program} returns the matching $\scr{M}$
when presented $V$ based on $(Y_v)_{v \in V}$ generated $u.a.r.$ from $[0,1]$.
In this case, if the variables $((\alpha^*_u)_{u \in U}, (\phi^*_{v,R})_{v \in V, R \subseteq U})$ 
are defined through the above charging scheme, then
\[
	\mb{E}[ w( \scr{M})] = F \cdot \left( \sum_{u \in U} \alpha^{*}_{u}  +  \sum_{v \in V} \sum_{R \subseteq U} \OPT(v,R) \cdot \phi_{v,R}^{*} \right).
\]
\end{lemma}
We also make the following claim regarding the feasibility of the variables $((\alpha^*_u)_{u \in U}, (\phi^*_{v,R})_{v \in V, R \subseteq U})$:

\begin{lemma} \label{lem:dual_feasibility_constraints_general}
If $G=(U,V,E)$ is a rankable stochastic graph
whose online nodes are presented to Algorithm \ref{alg:dynamical_program}
based on $(Y_v)_{v \in V}$ generated $u.a.r.$ from $[0,1]$,
then the solution $((\alpha^*_u)_{u \in U}, (\phi^*_{v,R})_{v \in V, R \subseteq U})$ 
is a feasible solution  to \ref{LP:DP_dual}.

\end{lemma}  
Since \ref{LP:DP} is a relaxation of the committal benchmark,
Theorem \ref{thm:ROM_rankable} follows from Lemmas \ref{lem:dual_variables_expected_value}
and \ref{lem:dual_feasibility_constraints_general} in conjunction with weak duality.
On the other hand, if we redefine $g(z):=1/2$ and $F:=1/2$, then analogous versions
of Lemmas \ref{lem:dual_variables_expected_value} and \ref{lem:dual_feasibility_constraints_general} hold,
even when the values $(Y_v)_{v\in V}$ are generated adversarially and $G$ is not rankable. With these
analogous lemmas, Theorem \ref{thm:adversarial} follows in the same way.
We focus on the ROM setting for the remainder of the section, as the analogous 
version of Lemma \ref{lem:dual_feasibility_constraints_general} for the adversarial
setting follows similarly, and is in fact easier to prove.

\subsubsection{Proving Dual Feasibility: Lemma \ref{lem:dual_feasibility_constraints_general}}

Let us suppose that the 
variables $((\alpha_u)_{u \in U}, (\phi_{v,R})_{v \in V, R \subseteq U})$ 
are defined as in the charging scheme of Section \ref{sec:charging_scheme}.
In order to prove Lemma \ref{lem:dual_feasibility_constraints_general},
we must show that for each fixed $u_0 \in U$ and $v_0 \in V$, we have that
\begin{equation}\label{eqn:dual_feasibility_fixed_vertices}
	\mb{E}[ p_{u_0,v_0} \cdot \alpha_{u_0} +   w_{u_0} \cdot p_{u_0,v_0} \, \sum_{\substack{R \subseteq U: \\ u_0 \in R}} \phi_{v,R}] \ge w_{u_0} \cdot p_{u_0,v_0}.
\end{equation}
Our strategy for proving \eqref{eqn:dual_feasibility_fixed_vertices}
first involves the same trick used by Devanur et al. \cite{DJK2013}.
Specifically, we define the stochastic graph $\til{G}:=(U, \til{V}, \til{E})$,
where $\til{V}:= V \setminus \{v_0\}$ and $\til{G}:=G[U \cup \til{V}]$. We wish to compare the execution of the algorithm on the instance $\til{G}$ to its execution on the instance $G$. It will be convenient to couple the randomness
between these two executions by making the following assumptions:
\begin{enumerate} \label{eqn:greedy_algorithm_coupling_general}
\item For each $e \in \til{E}$, $e$ is active in $\til{G}$ if and only if it is active in $G$.
\item The same random variables, $(Y_{v})_{v \in \til{V}}$, are used in both executions.
\end{enumerate}
If we now focus on the execution of $\til{G}$, then define the random variable
$\til{Y}_c$ where $\til{Y}_c:=Y_{v_c}$ 
if $u_0$ is matched to some $v_c \in \til{V}$, 
and $\til{Y}_c:=1$ if $u_0$ remains unmatched after the execution on $\til{G}$.
We refer to the random variable $\til{Y}_c$ as the \textbf{critical time} of
vertex $u_0$ with respect to $v_0$. We claim the following lower bound on $\alpha_{u_0}$ in terms of the critical time $\til{Y}_{c}$. We emphasize that this is the only part of the proof of Theorem \ref{thm:ROM_rankable} which requires the rankability
of $G$.
\begin{proposition} \label{prop:offline_dual_variable_lower_bound_general}
If $G$ is rankable, then $\alpha_{u_0} \ge \frac{w_{u_0}}{F} \, (1 - g(\til{Y}_c))$. 
\end{proposition}

\begin{proof}[Proof of Proposition \ref{prop:offline_dual_variable_lower_bound_general}]
For each $v \in V$, denote $R^{\text{af}}_{v}(G)$ as the unmatched (remaining) vertices of
$U$ right after $v$ is processed (attempts its probes) in the execution on $G$. We emphasize that if a probe of $v$
yields an active edge, thus matching $v$, then this match is excluded from $R^{\text{af}}_{v}(G)$.
Similarly, define $R^{\text{af}}_{v}(\til{G})$ in the same way for the execution on $\til{G}$
(where $v$ is now restricted to $\til{V}$).

Now, since $G$ is rankable and the constraints $(\scr{C}_v)_{v \in V}$ are substring-closed,
we can use the coupling between the two executions to inductively prove that
\begin{equation}\label{eqn:monotonicity_general}
	R^{\text{af}}_{v}(G) \subseteq R^{\text{af}}_{v}(\til{G}),
\end{equation}
for each $v \in \til{V}$ \footnote{Example \ref{example:bad_OPT_behaviour} shows why \eqref{eqn:monotonicity_general} will not hold if $G$ is not rankable.}. Now, since $g(1)=1$ (by assumption), there is nothing to prove if
$\til{Y}_{c}=1$. Thus, we may assume that $\til{Y}_c < 1$, and as a consequence,
that there exists some vertex $v_{c} \in V$ which matches to $u_0$ at time $\til{Y}_c$
in the execution on $\til{G}$.

On the other hand, by assumption we know that $u_0 \notin R^{\text{af}}_{v_c}(\til{G})$
and thus by \eqref{eqn:monotonicity_general}, that $u_0 \notin R^{\text{af}}_{v_c}(G)$.
As such, there exists some $v' \in V$ which probes $(u_0,v')$ and succeeds
in matching to $u_0$ at time $Y_{v'} \le \til{Y}_c$. Thus,
since $g$ is monotone,
\[
	\alpha_{u_0}  \ge \frac{w_{u_0}}{F} \, (1 - g(Y_{v'})) \, \bm{1}_{[ \til{Y}_c < 1]} \ge  \frac{w_{u_0}}{F} \, (1 - g(\til{Y}_c)),
\]
and so the claim holds.

\end{proof}
By taking the appropriate conditional expectation, we can also 
lower bound the random variables $(\phi_{v_{0},R})_{\substack{R \subseteq U: \\ u_{0} \in R} }$.
\begin{proposition} \label{prop:online_dual_variable_lower_bound_general}
\[\sum_{\substack{R \subseteq U: \\ u_{0} \in R}} \mb{E}[ \phi_{v_{0},R} \, | \, (Y_{v})_{v \in \til{V}}, (\sta(e))_{e \in \til{E}}] \ge  \frac{1}{F} \int_{0}^{\til{Y}_c} g(z) \, dz.\]
\end{proposition}

\begin{proof}[Proof of Proposition \ref{prop:online_dual_variable_lower_bound_general}]
We first define $R_{v_0}$ as
the unmatched vertices of $U$ when $v_0$ arrives (this is a random subset of $U$).
We also once again use $\scr{M}$ to denote the matching returned by Algorithm \ref{alg:dynamical_program}
when executing on $G$. If we now take a \textit{fixed} subset $R \subseteq U$, then
the charging assignment to $\phi_{v_0,R}$
ensures that
\[
	\phi_{v_0,R} = w( \scr{M}(v_0)) \cdot \frac{g(Y_{v_0})}{F \cdot \OPT(v_0,R)} \cdot \bm{1}_{[ R_{v_0} = R]},
\]
where $w( \scr{M}(v_0))$ corresponds to the weight of
the vertex matched to $v_0$ (which is zero if $v_0$ remains unmatched after the execution on $G$).
In order to make use of this relation,
let us first condition on the values of $(Y_{v})_{v \in V}$, as well as the 
states of the edges of $\til{E}$; that is, $(\sta(e))_{e \in \til{E}}$. Observe
that once we condition on this information, we can determine $g(Y_{v_0})$,
as well as $R_{v_0}$. As such,
\[
	\mb{E}[ \phi_{v_0, R} \, | \, (Y_{v})_{v \in V}, (\sta(e))_{e \in \til{E}}] = \frac{g(Y_{v_0})}{F \cdot \OPT(v_0,R)} \, \mb{E}[w( \scr{M}(v_0)) \, | \, (Y_{v})_{v \in V}, (\sta(e))_{e \in \til{E}}] \cdot  \bm{1}_{[ R_{v_0} = R]}.
\]
On the other hand, the only randomness which remains in 
the conditional expectation involving $w(\scr{M}(v_0))$ is 
over the states of the edges adjacent to $v_0$. Observe now
that since Algorithm \ref{alg:dynamical_program} behaves optimally
on $G[ \{v_{0}\} \cup R_{v_{0}}]$,
we get that
\begin{equation} \label{eqn:DP_optimality}
	\mb{E}[w( \scr{M}(v_0)) \, | \, (Y_{v})_{v \in V}, (\sta(e))_{e \in \til{E}}] = \OPT(v_0,R_{v_0}),
\end{equation}
and so for the \textit{fixed} subset $R \subseteq U$,
\[
	\mb{E}[w( \scr{M}(v_0)) \, | \, (Y_{v})_{v \in V}, (\sta(e))_{e \in \til{E}}] \cdot \bm{1}_{[ R_{v_0} = R]} = \OPT(v_0,R) \cdot\bm{1}_{[ R_{v_0} = R]}
\]
after multiplying each side of \eqref{eqn:DP_optimality} by the indicator random variable $\bm{1}_{[R_{v_{0}}=R]}$.
Thus, 
\[
	\mb{E}[\phi_{v_0,R}\, | \, (Y_{v})_{v \in V}, (\sta(e))_{e \in \til{E}}] = \frac{g(Y_{v_0})}{F} \, \bm{1}_{[ R_{v_0} = R]},
\]
after cancellation. We therefore get that
\[
	\sum_{\substack{R \subseteq U: \\ u_{0} \in R}} \mb{E}[ \phi_{v_{0},R} \, | \, (Y_{v})_{v \in V}, (\sta(e))_{e \in \til{E}}] 
	= \frac{g(Y_{v_0})}{F} \sum_{\substack{R \subseteq U: \\ u_{0} \in R}} \bm{1}_{[ R_{v_0} = R]}.
\]
Let us now focus on the case when $v_0$ arrives before the critical time; that is, $0 \le Y_{v_0} < \til{Y}_c$. Up until the arrival of $v_0$, the executions of the algorithm on $\til{G}$ and $G$ proceed identically, thanks to the coupling between the executions.
As such, $u_0$ must be available when $v_0$ arrives. 
We interpret this observation in the above notation as saying the following:
\[
	\bm{1}_{[ Y_{v_0} < \til{Y}_c]} \le \sum_{\substack{R \subseteq U: \\ u_{0} \in R}} \bm{1}_{[ R_{v_0} = R]}.
\]
As a result,
\[\sum_{\substack{R \subseteq U: \\ u_{0} \in R}} \mb{E}[ \phi_{v_{0},R} \, | \, (Y_{v})_{v \in V}, (\sta(e))_{e \in \til{E}}]
	\ge \frac{g(Y_{v_0})}{F} \, \bm{1}_{[ Y_{v_0} < \til{Y}_c]}.
\]
Now, if we take expectation over $Y_{v_0}$, while still conditioning on the random variables $(Y_{v})_{v \in \til{V}}$, then we get
that
\[
\mb{E}[ g(Y_{v_0}) \cdot \bm{1}_{[ Y_{v_0} < \til{Y}_c]} \, | \, (Y_{v})_{v \in \til{V}}, (\sta(e))_{e \in \til{E}}]= \int_{0}^{\til{Y}_c} g(z) \, dz,
\]
as $Y_{v_0}$ is drawn uniformly from $[0,1]$, independently from $(Y_{v})_{v \in \til{V}}$ and $(\sta(e))_{e \in \til{E}}$.
Thus, after applying the law of iterated expectations,
\[
	\sum_{\substack{R \subseteq U: \\ u_{0} \in R}} \mb{E}[ \phi_{v_{0},R} \, | \, (Y_{v})_{v \in \til{V}}, (\sta(e))_{e \in \til{E}}] \ge  \frac{1}{F} \int_{0}^{\til{Y}_c} g(z) \, dz,
\]
and so the claim holds.

\end{proof}

With Propositions \ref{prop:offline_dual_variable_lower_bound_general} and \ref{prop:online_dual_variable_lower_bound_general},
the proof of Lemma \ref{lem:dual_feasibility_constraints_general} follows easily:
\begin{proof}[Proof of Lemma \ref{lem:dual_feasibility_constraints_general}]
We first observe that by taking the appropriate conditional expectation, 
Proposition \ref{prop:offline_dual_variable_lower_bound_general} ensures that
\[
	\mb{E}[\alpha_{u_0} \, | \, (Y_{v})_{v \in \til{V}}, (\sta(e))_{e \in \til{E}}] \ge \frac{w_{u_0}}{F} \cdot (1 - g(\til{Y}_c)),
\]
where the right-hand side follows since $\til{Y}_c$ is entirely determined from $(Y_{v})_{v \in \til{V}}$ and $(\sta(e))_{e \in \til{E}}$. Thus, combined with Proposition \ref{prop:online_dual_variable_lower_bound_general},
\[
	\mb{E}[ p_{u_0,v_0} \cdot \alpha_{u_0}+   w_{u_0} \cdot p_{u_0,v_0} \cdot \sum_{\substack{R \subseteq U: \\ u_0 \in R}} \phi_{v,R} \, | \, (Y_{v})_{v \in \til{V}}, (\sta(e))_{e \in \til{E}}],
\]
is lower bounded by
\[
	\frac{w_{u_0} \cdot p_{u_0,v_0}}{F} \cdot ( 1 - g( \til{Y}_c)) + \frac{w_{u_0} \, p_{u_0,v_0}}{F} \int_{0}^{\til{Y}_c} g(z) \, dz .
\]
However, $g(z):= \exp(z-1)$ for $z \in [0,1]$ by assumption, so
\[
	( 1 - g( \til{Y}_c)) + \int_{0}^{\til{Y}_c} g(z) \, dz = \left(1 - \frac{1}{e} \right),
\]
no matter the value of the critical time $\til{Y}_c$. As such, since $F:= 1 - 1/e$,
we may apply the law of iterated expectations and conclude that
\[
\mb{E}[p_{u_0,v_0} \cdot \alpha_{u_0} + w_{u_0} \cdot p_{u_0,v_0} \cdot \sum_{\substack{R \subseteq U: \\ u_0 \in R}} \phi_{v,R}] \ge w_{u_0} \cdot p_{u_0,v_0}.
\]
As the vertices $u_{0} \in U$ and $v_{0} \in V$ were chosen arbitrarily, the
proposed dual solution of Lemma \ref{lem:dual_feasibility_constraints_general} is feasible, 
and so the proof is complete.

\end{proof}

\section{Conclusion and Open Problems} \label{sec:open_problems} 

We considered the online  stochastic bipartite matching with commitment in a number of different settings establishing several competitive bounds against the committal benchmark. In Appendix \ref{sec:non_committal_benchmark}, we indicate when our results hold against a stronger non-committal benchmark.  

In the case of vertex-weighted stochastic graphs, adversarial arrivals, and general probing constraints, we provide a deterministic algorithm that achieves a $\frac{1}{2}$ competitive ratio. 
This is an optimal competitive ratio for deterministic algorithms and adversarial arrivals.

In the case of the random order model, we provide two results. First, for edge weighted stochastic graphs, and general probing constraints, we provide a randomized algorithm that achieves the optimal asymptotic competitive ratio of $\frac{1}{e}$. 
For vertex weighted results we provide an algorithm that achieves a $1-1/e$ competitive ratio whenever the input graph is ``rankable''. Our rankable assumption subsumes most of the stochastic graph settings studied in previous works.

Our work leaves open a number of challenging open problems.
For context, we note that currently, even for the classical (i.e., non-stochastic) setting, $1-\frac{1}{e}$ is the best known ratio for deterministic algorithms operating on unweighted or  vertex weighted graphs  with random vertex arrivals. The best known ROM in-approximation of $0.823$ (due to Manshadi et al. \cite{ManshadiGS12}) comes from the classical i.i.d. unweighted graph setting for a known distribution and applies to randomized as well as deterministic algorithms.
\begin{itemize}
\item What is the best ratio that a deterministic or randomized  online algorithm can obtain for {\it all} stochastic graphs in the ROM setting? That is, what competitive ratio can be achieved without the rankable assumption?
Is there an online probing algorithm which can surpass the $1-1/e$ ``barrier''? In \cite{borodin2021prophet}, we show that $1-1/e$ is a hardness result for \textit{non-adaptive} online probing algorithms, even for the unweighted unit patience setting
when the stochastic graph is known to the algorithm.

\item Is there a provable difference between what an optimal online algorithm can obtain against 
the committal benchmark versus the non-committal benchmark? Specifically, does Algorithm \ref{alg:dynamical_program}
achieve a competitive ratio of $1-1/e$ against the non-committal benchmark which holds
for all rankable stochastic graphs or for all stochastic graphs? The hardness
result of Proposition \ref{prop:negative_online_node} suggests that Algorithm \ref{alg:dynamical_program} does \textit{not} attain a competitive ratio of $1-1/e$, even for rankable stochastic graphs.

\item What is the best ratio that a randomized online algorithm can obtain for stochastic graphs in the adversarial arrival model? The Mehta and Panigraphi \cite{MehtaP12} $0.621$ inapproximation shows  that randomized probing algorithms (even for unweighted graphs and unit patience) cannot achieve a $1-1/e$ performance guarantee against \ref{LP:standard_benchmark},
however the work of Goyal and Udwani \cite{Goyal2020OnlineMW} suggests that this is because \ref{LP:standard_benchmark}
is too loose a relaxation of the committal benchmark.

\item Is there a online stochastic matching problem in which
the optimum competitive ratio provably worse than the optimal ratio for the corresponding classical setting? Note that in the classical setting the benchmark is the weight of an offline optimal matching. 

\item Can our $1-\frac{1}{e}$ competitive ratio be improved by a randomized algorithm in the vertex-weighted ROM setting? Here we note that in the classical ROM setting, the \textsc{Ranking} algorithm achieves a $0.696$ ratio for unweighted graphs  (due to Mahdian and Yan \cite{Mahdian2011}) and a $0.6534$ ratio (due to Huang et al. \cite{huang2018online}) for vertex weighted graphs. Thus, randomization seems  to significantly  help in the classical ROM setting.

\end{itemize}

\bibliographystyle{plain}
\bibliography{bibliography}

\appendix

\section{The Non-committal Benchmark} \label{sec:non_committal_benchmark}

In this section, we introduce the \textbf{non-committal benchmark}.
This benchmark must still adaptively probe edges subject to probing constraints, and its goal
is the same as the committal benchmark, but it does not need to respect
commitment. More precisely, if $P_a \subseteq E$ corresponds to
the active probes made by the benchmark, then
it returns a matching $\scr{M} \subseteq P_a$ of maximum weight.
We denote $\nOPT(G)$ as the expected weight of the matching that the non-committal benchmark constructs,
and abuse notation slightly by also using $\nOPT(G)$ to refer to the \textit{strategy}
of the non-committal benchmark on $G$. Observe that in the case of unlimited patience, $\nOPT(G)$ may probe all the edges of $G$,
and thus corresponds to the expected weight of the optimum matching of the stochastic graph.
Clearly, for any set of probing constraints, the non-committal benchmark is no weaker than the committal benchmark;
that is, $\nOPT(G) \ge \OPT(G)$ for any stochastic graph $G$. We first show that these values are separated by a ratio of at least $0.856269$, even for a single online node.
\begin{proposition}\label{prop:negative_online_node}
There exists a stochastic graph $G$ with a single online node $v$, such
that $\OPT(G) = 0.856269 \cdot \nOPT(G)$. 
\end{proposition}
\begin{proof}[Proof of Proposition \ref{prop:negative_online_node}]
Suppose $G$ has a single online node $v$ with patience $\ell_{v} =2$, and that there are $3$ offline nodes $U=\{u_{1}, u_{2}, u_{3}\}$. For each $i \in \{1,2,3\}$, we denote the weight of $(u_{i},v)$ by $w_{i}$ and assume that the edge $(u_i,v)$
is active with probability $p_{i}$. We make the following assumptions on these weights
and probabilities:
\begin{enumerate}
\item $w_{1} < w_{2} < w_{3}$.
\item $p_{1} > p_{2} > p_{3}$.
\item $w_{1} \cdot p_{1} \ge w_{2} \cdot p_{2} > w_{3} \cdot p_{3}$. \label{eqn:expected_reward}
\item $p_{2} \cdot w_{2} - p_{3} \cdot w_{3} \ge p_{1} \cdot w_{1} \cdot (p_{2} - p_{3})$ \label{eqn:OPT_strategy}
\end{enumerate}
Clearly, there exists a choice of weights and probabilities which satisfy these constraints. For instance,
take $w_{1} =3$, $w_{2} =4$, $w_{3}= 98$, $p_{1}=0.8$, $p_{2}=0.6$, and $p_{3} =0.01$.
Based on these assumptions let us now consider the value of $\OPT(G)$. Observe that
\[
	p_{2} \cdot w_{2} + (1 - p_{2}) \cdot p_{1} \cdot w_{1} \ge p_{3} \cdot w_{3} + (1- p_{3}) \cdot p_{1} \cdot w_{1}
															\ge p_{3} \cdot w_{3} + (1 - p_{3}) \cdot p_{2} \cdot w_{2},
\]
where the first inequality follows from \eqref{eqn:OPT_strategy}, and the second follows from \eqref{eqn:expected_reward}.
As a result, it is clear to see that the committal benchmark corresponds to probing $(u_{2},v)$ and then $(u_{1},v)$ (if necessary);
thus, $\OPT(G) = p_{2} \cdot w_{2} + (1 - p_{2}) \cdot p_{1} \cdot w_{1}$. On the other hand, let us consider $\nOPT(G)$, the value of the non-committal benchmark on $G$.
Consider the following \textit{non-committal} probing algorithm:
\begin{itemize}
\item Probe $(u_{2},v)$, and if $\sta(u_{2},v)=1$, probe $(u_{3},v)$. \label{eqn:adaptive_step}
\item Else if $\sta(u_{2},v)=0$, probe $(u_{1},v)$.
\item Return the edge of highest weight which is active (if any).
\end{itemize}
Clearly, this probing algorithm uses adaptivity to decide whether to reveal $(u_{3},v)$ or
$(u_{1},v)$ in its second probe. Specifically, if it discovers that $(u_{2},v)$ is active,
then it knows that it will return an edge with weight at least $w_{2}$. As such, it only
makes sense for $(u_{3},v)$ to be its next probe, as $w_{3} > w_{2} > w_{1}$.
On the other hand, if $(u_{2},v)$ is discovered to be inactive, it makes sense to
prioritize probing the edge $(u_{1},v)$ over $(u_{3},v)$, as the \textit{expected}
reward is higher; namely, $w_{1} \cdot p_{1} > w_{3} \cdot p_{3}$.
Thus, it is clear that the expected weight of the edge returned is 
\[
p_{2} \cdot p_{3} \cdot w_{3} + p_{2} \cdot (1 - p_{3}) \cdot w_{2} + (1 - p_{2}) \cdot p_{1} \cdot w_{1}.
\] 
Observe however that
\begin{align*}
	p_{2} \cdot p_{3} \cdot w_{3} + p_{2} \cdot (1 - p_{3}) \cdot w_{2} + (1 - p_{2}) \cdot p_{1} \cdot w_{1} 
				  &> p_{2} \cdot w_{2} + (1 - p_{2}) \cdot p_{1} \cdot w_{1}	\\
				  &= \OPT(G),
\end{align*}
where the final inequality follows since $w_{3} > w_{2}$.
As a result, it is clear that this strategy corresponds to the non-committal benchmark, and so 
$\OPT_{non}(G) > \OPT(G)$. In fact,
for the specific choice when $w_{1} =3$, $w_{2} =4$, $w_{3}= 98$, $p_{1}=0.8$, $p_{2}=0.6$, and $p_{3} =0.01$,
it holds that $\OPT(G) = 0.856269 \cdot \nOPT(G)$. 
\end{proof}
This example slightly improves upon the negative result of \cite{costello2012matching}, in which Costello et al. show
that the ratio between $\OPT(G)$ and $\nOPT(G)$ is at most $0.898$ (albeit for unweighted graphs).  

We remark that restricted to a single online node, the non-committal probing problem is a special case of
the \textit{adaptive} version of \textsc{ProblemMax}, a stochastic probing problem which is studied in \cite{Asadpour2016, Fu2018, Segev2020}. Similarly, the committal probing problem is a special case of the \textit{non-adaptive} version of
\textsc{ProblemMax}, which is also considered in \cite{Asadpour2016, Fu2018, Segev2020}. Thus, we can view $0.856269$
as an upper bound (negative result) on the \textbf{adaptivity gap} of \textsc{ProblemMax}. This is in contrast to the
lower bound of $1-1/e$ on the adaptivity gap of \textsc{ProblemMax} proven by Asadpour et al. \cite{Asadpour2016}. 

We now provide a number of settings
in which the competitive ratios of Algorithm \ref{alg:dynamical_program} holds
against the non-committal benchmark.

\begin{theorem}\label{thm:non_committal_competitive}
Let $G$ be a vertex-weighted stochastic graph such that
for each $v \in V$ and $R \subseteq U$, $\OPT(v,R) =\nOPT(v,R)$.
In this case, Algorithm \ref{alg:dynamical_program} attains a performance
guarantee of $1/2$ against $\nOPT(G)$ assuming adversarial arrivals,
and $1-1/e$ against $\nOPT(G)$ assuming random order arrivals
and the rankability of $G$.
\end{theorem}
\begin{remark}
If $G$ is unweighted or is vertex-weighted yet
has unit/unlimited patience, then Theorem \ref{thm:non_committal_competitive}
applies. 
\end{remark}

To prove Theorem \ref{thm:non_committal_competitive}, consider the following LP, 
where each edge $e \in E$ of the edge weighted
stochastic graph $G=(U,V,E)$
is associated with two variables, namely $x_{e}$ and $z_{e}$. We interpret the former variable as the probability
that the non-committal benchmark probes the edge $e$, whereas the latter
variable corresponds to the probability that $e$ is included
in the matching constructed by the non-committal benchmark. Note that
for convenience, we assume that $E=U \times V$.
\begin{align}\label{LP:DP-non}
\tag{LP-DP-non}
& \text{maximize}  & \sum_{u \in U, v \in V} w_{u,v} \cdot z_{u,v} \\
&\text{subject to} & \sum_{v \in V}  z_{u,v} &\leq 1 &&\forall u \in U\\
&  &\sum_{u \in R} w_{u,v} \cdot z_{u,v} & \le  \nOPT(v,R)  && \forall v \in V, \, R \subseteq U \\ \label{eqn:OPT_non-equiv}
&  &z_{u,v} &\le p_{u,v} \cdot x_{u,v} && \forall u \in U, v \in V	\\  
&  &x_{u,v},z_{u,v} &\ge 0 && \forall u \in U, v \in V
\end{align}
We denote $\nLPOPT(G)$ as the value of an optimum solution to \ref{LP:DP-non}.
\begin{lemma}\label{lem:non_committal_relaxation}
For any stochastic graph $G=(U,V,E)$ with substring-closed probing constraints
$(\scr{C}_v)_{v \in V}$, $\nOPT(G) \le \nLPOPT(G)$.
\end{lemma}

\begin{proof}
Let us suppose that $\scr{M}$ is the matching returned by the non-committal benchmark when it executes on $G=(U,V,E)$.
If we fix $u \in U$ and $v \in V$, then define $x_{u,v}$ as the probability the non-committal benchmark probes the edge $(u,v)$, and $z_{u,v}$ as the probability that it includes $e$ in $\scr{M}$. Observe then that
\[
	\nOPT(G) = \mb{E}[w(\scr{M}) ] = \sum_{u \in U, v \in V} w_{u,v} \cdot z_{u,v}.
\]
Thus, we need only show that $(x_{u,v},z_{u,v})_{u \in U,v \in V}$
is a feasible solution to \ref{LP:DP-non}, as this will ensure that
$\nOPT(G) = \sum_{u \in U, v \in V} w_{u,v} \cdot z_{u,v} \le \nLPOPT(G)$.
If we first fix $u \in U$ and $v \in V$, then observe that
in order for $(u,v)$ to be included in $\scr{M}$, $(u,v)$ must be probed \textit{and} $(u,v)$ must be active. On the other
hand, these two events occur independently of each other. As such, $z_{u,v} \le p_{u,v} \cdot x_{u,v}$.
Now, each $u \in U$ is matched at most once by the non-committal benchmark,
thus $\sum_{v \in V} z_{u,v} \le 1.$
Finally, fix $v \in V$, and denote $\scr{M}(v)$ as the edge
matched to $v$ (which is $\emptyset$ by convention if $v$ remains unmatched),
and denote $w(\scr{M}(v))$ as the weight of the edge $v$ is matched
to (which is $0$ provided $v$ remains unmatched).
Observe first that $\sum_{u \in U} w_{u,v} \cdot z_{u,v} = \mb{E}[ w( \scr{M}(v))]$.
Moreover, executing the non-committal benchmark on $G$ induces\footnote{The strategy $\scr{B}_{v}$ can be defined
formally by first drawing (simulated) independent copies of the edge states which are not adjacent to $v$,
say $\til{\sta}(e)_{e \in E: v \notin e}$. By executing the non-committal benchmark
on $G$ with $\til{\sta}(e)_{e \in E: v \notin e}$ and $\sta(e)_{e \in \partial(v)}$, we get the desired strategy on $G[\{v\} \cup U]$.}
a probing strategy on $G[\{v\} \cup U]$, which we denote by $\scr{B}_{v}$.
However, observe that since the non-committal benchmark decides upon which edges
to match after all its probes are made, so does $\scr{B}_v$. Specifically, the
match it makes to $v$ is determined once all its probes to $U \times \{v\}$ are made. 
Clearly, the expected value of this match is equal to
$\mb{E}[w( \scr{M}(v))]$ and can be no larger than $\nOPT(v,U)$.
Thus,
\[
	 \sum_{u \in U} w_{u,v} \cdot z_{u,v} = \mb{E}[ w( \scr{M}(v))] \le \nOPT(v,U).
\]
More generally, if we fix $R \subseteq U$, then
\[
	\sum_{u \in U} w_{u,v} \cdot z_{u,v} = \mb{E}[ w(\scr{M}(v)) \cdot \bm{1}_{[\scr{M}(v) \in R \times \{v\}]}] \le \nOPT(v,R).
\]
To see this, consider a modification of $\scr{B}_{v}$, say $\scr{B}_{v}(R)$, which matches $v$ to $u \in U$
if and only if $\scr{B}_v$ matches $v$ to $u$ \textit{and} $(u,v) \in R \times \{v\}$.

This shows that all the constraints of \ref{LP:DP-non} hold
for $(x_{u,v},z_{u,v})_{u \in U, v \in V}$, and so the proof is complete.
\end{proof}

\begin{proof}[Proof of Theorem \ref{thm:non_committal_competitive}]
Let us suppose that $G$ is vertex-weighted.

Consider a modified version of \ref{LP:DP} 
in which the right-hand side of
constraint \eqref{eqn:OPT_constraint} is replaced
by the analogous expression for the non-committal benchmark (i.e.,
$\nOPT(v,R)$). It is not hard to show that this modified
LP is a reformulation of \ref{LP:DP-non}. Thus, because of the assumptions
on $G$, $\dLPOPT(G) = \nLPOPT(G)$. Since the performance
guarantees of Algorithm \ref{alg:dynamical_program} in Theorems \ref{thm:adversarial}
and \ref{thm:ROM_rankable} are proven
against $\dLPOPT(G)$, Theorem \ref{thm:non_committal_competitive}
follows thanks to Lemma \ref{lem:non_committal_relaxation}.

\end{proof}

\section{Relaxing the Committal Benchmark} \label{sec:committal_relaxation}

\label{sec:LP_relaxation}
Suppose that we are given an arbitrary stochastic
graph $G=(U,V,E)$. Let us restate \ref{LP:config} 
for convenience:
\begin{align}\label{LP:config_restatement}
\tag{LP-config}
&\text{maximize} &  \sum_{v \in V} \sum_{\bm{e} \in \scr{C}_v } \val(\bm{e}) \cdot x_{v}(\bm{e}) \\
&\text{subject to} & \sum_{v \in V} \sum_{\substack{ \bm{e} \in \scr{C}_v: \\ (u,v) \in \bm{e}}} 
p_{u,v} \cdot g(\bm{e}_{< (u,v)}) \cdot x_{v}( \bm{e})  \leq 1 && \forall u \in U  \label{eqn:relaxation_efficiency_matching_restatement}\\
&& \sum_{\bm{e} \in \scr{C}_v} x_{v}(\bm{e}) = 1 && \forall v \in V,  \label{eqn:relaxation_efficiency_distribution_restatement} \\
&&x_{v}( \bm{e}) \ge 0 && \forall v \in V, \bm{e} \in \scr{C}_v
\end{align}
We contrast \ref{LP:config_restatement} with \ref{LP:standard_definition_general}, which is defined
only when $G$ has patience values $(\ell_v)_{v \in V}$:
\begin{align}\label{LP:standard_definition_general}
\tag{LP-std}
&\text{maximize} & \sum_{e \in E} w_{e} \cdot p_{e} \cdot x_{e} \\
&\text{subject to} & \sum_{e \partial(u)} p_{e} \cdot x_{e} & \leq 1 && \forall u \in U \\
& &\sum_{e \in \partial(v)} p_{e} \cdot x_{e} & \leq 1 && \forall v \in V  \\
& &\sum_{e \in \partial(v)} x_{e} & \leq \ell_v && \forall v \in V  \\
& &0 \leq x_{e} &\leq 1 && \forall e \in E.
\end{align}
Observe that \ref{LP:config} and \ref{LP:standard_definition_general} are the same
LP in the case of unit patience:
\begin{align} \label{LP:standard_benchmark}
\tag{LP-std-unit}
&\text{maximize}&  \sum_{v \in V} \sum_{e \in \partial(v)} w_e \cdot p_{e} \cdot x_{e}\\
&\text{subject to}& \sum_{e \in \partial(u)} p_{e} \cdot x_{e} &\leq 1 && \forall u \in U\\
&&\sum_{e \in \partial(v)} x_{e} &\leq 1 && \forall v \in V\\
&&  x_{e}  &\ge 0 && \forall e \in E
\end{align}

For completeness, we now present the essential ideas used in the proof of Theorem \ref{thm:new_LP_relaxation},
which shows that \ref{LP:config_restatement} relaxes the committal benchmark.

Let us suppose that hypothetically we could make the following assumption
regarding the committal benchmark:
\begin{enumerate}[label=(\subscript{P}{{\arabic*}})]
\item If $e=(u,v)$ is probed and $\st(e)=1$, then $e$ is included in the matching, provided $v$ is currently unmatched.
\label{eqn:single_vertex_committal}
\item For each $v \in V$, the edge probes involving $\partial(v)$ are made independently of the edge states $(\sta(e))_{e \in \partial(v)}$. \label{eqn:single_vertex_non_adaptivity}
\end{enumerate}
Observe then that \ref{eqn:single_vertex_committal} and \ref{eqn:single_vertex_non_adaptivity} would imply that
the expected weight of the edge assigned to $v$ is
$\sum_{\bm{e} \in \scr{C}_v } \val(\bm{e}) \cdot x_{v}(\bm{e})$. 
Moreover, the left-hand side of \eqref{eqn:relaxation_efficiency_matching} 
would correspond to the probability that $u \in U$ is matched,
so $(x_{v}(\bm{e}))_{v \in V, \bm{e} \in \scr{C}_v}$
would be a feasible solution to \ref{LP:config_restatement}, and so
we could upper bound $\OPT(G)$ by $\LPOPT(G)$. Now, if we knew that
the committal benchmark adhered to some adaptive vertex ordering $\pi$ on $V$
(i.e., it chooses $v_{\pi(i)}$ based on $v_{\pi(1)},\ldots v_{\pi(i-1)}$,
and probes the edges of $v_{\pi(i)}$ before moving to $v_{\pi(i+1)}$), then it is clear that we could assume \ref{eqn:single_vertex_committal} and \ref{eqn:single_vertex_non_adaptivity} simultaneously\footnote{It is clear that we may assume
the committal benchmark satisfies \ref{eqn:single_vertex_committal} $w.l.o.g.$, but not \ref{eqn:single_vertex_non_adaptivity}.}
$w.l.o.g.$. However, clearly a probing algorithm with this restriction is in general less
powerful than the committal benchmark. As such, the natural interpretation of the
variables of \ref{LP:config_restatement} does not seem to easily lend itself to a proof
of Theorem \ref{thm:new_LP_relaxation}.

In order to get around these issues, we first discuss the relaxed stochastic matching problem
defined in Section \ref{sec:edge_weights} in more detail.
A solution to this problem is a relaxed probing algorithm.
A relaxed probing algorithm operates on the stochastic graph $G=(U,V,E)$ in
the same framework as an offline probing algorithm.
That is, initially the edge states $(\st(e))_{e \in E}$ are unknown 
to the algorithm, and it must adaptively probe the edges
of $G$, while respecting the probing constraints of the online nodes of $G$.
Its output is then a subset of its probes which yielded
active edges, which we denote by $\scr{N}$. The goal of
the relaxed probing algorithm is to maximize the expected weight of $\scr{N}$,
while ensuring that the following properties are satisfied:
\begin{enumerate}
\item Each $v \in V$ appears in at most one edge of $\scr{N}$.
\item For each $u \in U$, the \textit{expected} number of edges of $\scr{N}$
which contain $u$ is at most one.
\end{enumerate}
We refer to $\scr{N}$ as a \textbf{one-sided matching} of the online nodes. We abuse
terminology slightly, and say that a vertex of $G$ is matched by $\scr{N}$, provided it is included
in an edge of $\scr{N}$. A relaxed probing algorithm must \textbf{respect commitment}. That is, it has the property that if a probe to $e=(u,v)$ yields an active edge, then the edge is included in $\scr{N}$ (provided $v$ is currently not in $\scr{N}$). Observe that this requires the relaxed probing algorithm to include $e$, even if $u$ is already adjacent to some
element of $\scr{N}$.

We define the \textbf{relaxed benchmark} as an optimum relaxed probing algorithm
on $G$, and denote $\rOPT(G)$ as the expected value of its output when executing
on $G$. Observe that by definition, $\OPT(G) \le \rOPT(G)$, where $\OPT(G)$ is the value of the committal benchmark on $G$.

Finally, we say that a relaxed probing algorithm is \textbf{non-adaptive},
provided its edge probes are statistically independent
from the edge states of $G$; that is, the random variables $(\sta(e))_{e \in E}$. We
emphasize that this is equivalent to (randomly) specifying an ordering $\pi$ on a subset of $E$, 
which satisfies the constraints $(\scr{C}_v)_{v \in V}$. The edges specified by $\pi$ are then probed in order, and an active edge
is added to the matching, provided its \textit{online} node is unmatched.

Unlike the committal benchmark, $\rOPT(G)$ can be attained
by a non-adaptive relaxed probing algorithm.

\begin{lemma}[Lemma $3.2$ in \cite{borodin2021prophet}] \label{lem:non_adaptive_optimum}
There exists a relaxed probing algorithm which is non-adaptive
and attains value $\rOPT(G)$ in expectation.
\end{lemma}

We refer the reader to \cite{borodin2021prophet} for a complete proof
of Lemma \ref{lem:non_adaptive_optimum},
and instead show how it allows us to prove Theorem \ref{thm:new_LP_relaxation}.
In fact, we prove that \ref{LP:config_restatement} encodes the value
of the relaxed benchmark exactly, thus implying Theorem \ref{thm:new_LP_relaxation}
since $\OPT(G) \le \rOPT(G)$. 
\begin{theorem}[Theorem $3.3$ in \cite{borodin2021prophet}]
For any stochastic graph $G$, an optimum solution to \ref{LP:config_restatement} has value
equal to $\rOPT(G)$, the value of the relaxed benchmark on $G$.
\end{theorem}
\begin{proof}
Suppose we are presented a solution $(x_{v}(\bm{e}))_{v \in V, \bm{e} \in \scr{C}_v}$
to \ref{LP:config_restatement}. We can then define the following relaxed probing algorithm:
\begin{enumerate}
\item $\scr{N} \leftarrow \emptyset$.
\item For each $v \in V$, set $e \leftarrow \textsc{VertexProbe}(v, \partial(v), (x_{v}(\bm{e}))_{\bm{e} \in \scr{C}_v}).$
If $e \neq \emptyset$, then add $e=(u,v)$ to $\scr{N}$, provided $v$ is currently unmatched.
\item Return $\scr{N}$.
\end{enumerate}
Using Lemma \ref{lem:fixed_vertex_probe}, it is clear that
$\mb{E}[ w(\scr{N})] = \sum_{v \in V} \sum_{\bm{e} \in \scr{C}_v} \val(\bm{e}) \cdot x_{v}(\bm{e})$.
Moreover, each vertex $u \in U$ is matched by $\scr{N}$ at most once in expectation, as a consequence of
\eqref{eqn:relaxation_efficiency_matching_restatement}.

In order to complete the proof, it remains to show that if $\scr{A}$ is
an optimum relaxed probing algorithm,
then there exists a solution to \ref{LP:config_restatement} whose value is equal to
$\mb{E}[w(\scr{A}(G))]$ (where $\scr{A}(G)$ is the one-sided matching returned by $\scr{A}$).
In fact, by Lemma \ref{lem:non_adaptive_optimum}, we may assume that $\scr{A}$
is non-adaptive.
Observe then that for each $v \in V$ and $\bm{e} =(e_{1}, \ldots ,e_{k}) \in \scr{C}_v$ with $k:=|\bm{e}| \ge 1$ we can define
\[
	x_{v}(\bm{e}):= \mb{P}[\text{$\scr{A}$ probes the edges $(e_i)_{i=1}^{k}$ in order}], 
\]
where $x_{v}(\lambda)$ corresponds to the probability no edge adjacent to $v$ is probed.
Setting $\scr{N}= \scr{A}(G)$ for convenience, observe that 
if $w(\scr{N}(v))$ corresponds to the weight of the edge
assigned to $v$ (which is $0$ if no assignment is made),
then
\[
	\mb{E}[ w(\scr{N}(v))] = \sum_{\bm{e} \in \scr{C}_v} \val(\bm{e}) \cdot x_{v}(\bm{e}),
\]
as $\scr{A}$ is non-adaptive. 
Moreover, for each $u \in U$,
\[
\sum_{v \in V} \sum_{\substack{ \bm{e} \in \scr{C}_v: \\ (u,v) \in \bm{e}}} 
p_{u,v} \cdot g(\bm{e}_{< (u,v)}) \cdot x_{v}( \bm{e})  \leq 1
\]
by once again using the non-adaptivity of $\scr{A}$. The proof is therefore complete.
\end{proof}

\end{document}